\documentclass[amsmath, aps, prx, twocolumn, notitlepage,10pt,longbibliographystyle]{revtex4-2}

\usepackage{graphicx}
\usepackage{dcolumn}
\usepackage{bm}
\usepackage[usenames]{color}
\usepackage{slashed}
\usepackage[utf8]{inputenc}
\usepackage{amsfonts}
\usepackage{amsmath}
\usepackage{amssymb}
\usepackage{amsthm}
\usepackage{hyperref}
\usepackage{latexsym}
\usepackage{enumerate}
\usepackage{bbm}
\usepackage{algpseudocode}
\usepackage{algorithm}

\makeatletter
\def\Let@{\def\\{\notag\math@cr}}
\makeatother

\newtheorem{theorem}{Theorem}

\begin{document}

\title{Recurrence method in Non-Hermitian Systems}
\author{Haoyan Chen}
\email{chenhaoyan@stu.pku.edu.cn}
\affiliation{International Center for Quantum Materials, School of Physics, Peking University, Beijing 100871, China}
\author{Yi Zhang}
\email{frankzhangyi@gmail.com}
\affiliation{International Center for Quantum Materials, School of Physics, Peking University, Beijing 100871, China}
\date{\today}

\begin{abstract}
We propose a novel and systematic recurrence method for the energy spectra of non-Hermitian systems under open boundary conditions based on the recurrence relations of their characteristic polynomials. Our formalism exhibits better accuracy and performance on multi-band non-Hermitian systems than numerical diagonalization or the non-Bloch band theory. It also provides a targeted and efficient formulation for the non-Hermitian edge spectra. As demonstrations, we derive general expressions for both the bulk and edge spectra of multi-band non-Hermitian models with nearest-neighbor hopping and under open boundary conditions, such as the non-Hermitian Su-Schrieffer-Heeger and Rice-Mele models and the non-Hermitian Hofstadter butterfly - 2D lattice models in the presence of non-reciprocity and perpendicular magnetic fields, which is only made possible by the significantly lower complexity of the recurrence method. In addition, we use the recurrence method to study non-Hermitian edge physics, including the size-parity effect and the stability of the topological edge modes against boundary perturbations. Our recurrence method offers a novel and favorable formalism to the intriguing physics of non-Hermitian systems under open boundary conditions. 
\end{abstract}

\maketitle

\section{Introduction} 

Non-Hermiticity is a significant factor that can not be overlooked in both classical and quantum systems~\cite{Hatano1996,Hatano1997,Hatano1998,Bender1998,Ruschhaupt2005,Bender2007,Hu2011,Regensburger2012,Yoshida2019,Schindler2011,Ezawa2019,Peng2014,Yuto2020,Hu2023}. In these systems, non-Hermiticity gives rise to various unconventional phenomena such as non-Hermitian skin effect~\cite{Kohei2020,Zhang2022,Song2019,Longhi2021,Claes2021,Yi2020,Zhang2020,Longhi2019,Shinsei2023,Li2023,Schindler2021,Marko2022,Wang2024,Li2020,Brandenbourger2019,Borgnia2020}, non-Bloch parity-time symmetry~\cite{Longhi2019Opt,Nie2024,Xiao2021,Hu2024}, non-Hermitian topology~\cite{Kohei2019symmetry,Gong2018,Okuma2023,Kawabata2022,Kunst2019,Xiao2022,Scheibner2020,Wanjura2021,Zhou2019,Liu2019,Yao2018,Yao2018Chern}, generalized Brillouin zone~\cite{Fu2023,Yao2018,Zhesen2020,Kazuki2019,Chen2023}, and exceptional points~\cite{Berry2004,Heiss2004,Bergholtz2021,Regensburger2013,Heiss2012,Dembowski2001}. One intriguing property of non-Hermitian systems is the breakdown of spectral stability under different boundary conditions~\cite{Xiong2018,Guo2021,Mao2021,Bottcher2005,Okuma2021}. Specifically, it is widely believed that the energy levels of Hermitian systems with an open boundary condition (OBC) can be reproduced mainly by the Bloch band theory under periodic boundary conditions (PBCs), even if the translation symmetry is broken. In non-Hermitian systems, however, the energy spectra under OBC change drastically and collapse into open arcs encircled by the loop-shape periodic boundary spectra~\cite{Okuma2020, Yao2018, Zhang2020, Wu2022}. It is widely known that these non-Hermitian OBC spectra can be obtained by the non-Bloch theory, which extends the Brillouin zone to the generalized Brillouin zone (GBZ) in non-Hermitian systems~\cite{Yao2018, Kazuki2019, Zhesen2020, Zhang2020}. Unfortunately, the determination of GBZ is somewhat challenging; therefore, it may be necessary to use the auxiliary GBZ technique~\cite{Zhesen2020}, which still involves additional concepts and may not apply to complicated multi-band non-Hermitian systems, such as non-Hermitian Hofstadter butterflies~\cite{Zhesen2020}. Furthermore, non-Bloch theory primarily emphasizes the bulk OBC spectra relevant to the GBZ, and specialized tools for edge spectra of non-Hermitian systems, intimately connected with topological properties, are still largely lacking. 

To circumvent the limitations of the GBZs and the non-Bloch band theory, we propose in this paper an alternative, general, efficient, and accurate method based on the recurrence relations of the eigenvalue equations in non-Hermitian systems under OBC. First, we conclude that the characteristic polynomial $D_N(E)=\det(E-H_N)$ of any non-Hermitian tight-binding Hamiltonian under OBC satisfies a recurrence relation with constant coefficients independent of the matrix size $N$. Then, we can represent $D_N(E)$ with the solutions $z_i(E)$ of the corresponding characteristic equation. The main result is that the OBC spectra include the energies such that (i) two maximal complex solutions $z_i$ and $z_j$ have the same modulus, or (ii) when representing $D_N$ by $z_i$, the maximal complex solution term is zero. We will argue that the above conditions are related to the bulk and edge parts of OBC spectra, respectively. 

As an application of the recurrence method, we derive a general expression for the bulk OBC spectra of non-Hermitian multi-band Hamiltonians with nearest-neighbor hoppings. This expression has the form of polynomial equations that allow us to analytically or numerically solve the OBC spectra without either diagonalization or determination of the GBZ. Thus, the recurrence method avoids typical accuracy problems in diagonalizing non-Hermitian matrices and significantly reduces the computational complexity. Furthermore, our formalism straightforwardly determines non-Hermitian Hamiltonians' edge spectra, which are unavailable from the conventional non-Bloch band theory. Based on this framework, we can analyze the effects of the overall number of sites (e.g., the size-parity effect)~\cite{Qi2023,Longhi2019topo,Chang2023} and the boundary perturbations~\cite{Li2023scale,Nakamura2023,Schindler2023,Landi2022} on the edge spectra. Such bulk and edge spectra offers direct information on physical properties such as the PT-symmetry breaking, topological characters, and bulk-boundary correspondence of the target non-Hermitian systems under OBCs. For example, we study the non-Hermitian Hofstadter model in two dimensions as an interplay between the non-reciprocity and magnetic field~\cite{Shao2022}. Our results provide robust evidence for the complex semiclassical theory~\cite{Yang2024} and topological characterizations of non-Hermitian systems~\cite{Kohei2019symmetry, Gong2018, Okuma2023, Kawabata2022, Kunst2019, Xiao2022, Scheibner2020, Wanjura2021, Zhou2019, Liu2019, Yao2018,Yao2018Chern}. Generally, our work offers a novel technique and perspective for non-Hermitian systems' OBC spectra and physics. 

We organize the rest of the paper as follows: In Sec.~\ref{sec:recurrence_method}, we discuss the recurrence relations of the characteristic polynomials for non-Hermitian tight-binding Hamiltonians under OBC. Based on a mathematical theorem for the solutions of the polynomials satisfying the recurrence relations, we provide a general formulation for the OBC spectra of non-Hermitian systems, which contains one part corresponding to the bulk spectra and the other corresponding to the edge spectra, respectively. In Sec.~\ref{sec:k_periodic}, we apply the recurrence method for the bulk spectra of non-Hermitian multi-band Hamiltonians with nearest-neighbor hoppings. We establish the general expressions on the bulk spectra for systems with an arbitrary number of sublattices and present the non-Hermitian Rice-Mele model as an example. In Sec.~\ref{sec:edge_spectra}, we demonstrate the advantages of the recurrence method over the edge spectra, continuing on the models in the previous section. Further, we investigate the effect of overall system sizes and boundary perturbations on edge spectra as modifications to the initial conditions of the recurrence relations, both of which are not accessible from the previous non-Bloch band theory. In Sec.~\ref{sec:nH_Hofstadter}, we apply our formalism non-Hermitian Hofstadter models on a triangular lattice. We explore the impact of non-reciprocity on the spectral patterns and validate the complex semiclassical theory in non-Hermitian systems. Besides, we examine the edge modes of the non-Hermitian Hofstadter models for their topological characters. We conclude our studies in Sec.~\ref{sec:conclusion} and discuss future outlooks and prospects.

\section{Recurrence Method in non-Hermitian systems} \label{sec:recurrence_method}

We first introduce a recurrence method for the spectra of generic non-Hermitian lattice models under OBCs. Consider a 1D non-Hermitian Hamiltonian with $N$ unit cells: 
\begin{equation}
    H_N=\sum\limits_{1\leq i,j\leq N;\alpha,\beta}t_{\alpha\beta}c_{i\alpha}^\dagger c_{j\beta}, \label{eq:HN}
\end{equation}
where $i,j$ represent different unit cells, and $\alpha,\beta \in [1, k]$ denote intracell sublattice degrees of freedom. The OBC spectrum of Eq.~\eqref{eq:HN} is determined by solutions of the following characteristic polynomial:
\begin{equation}
    D_{kN}(E)=\det(I_{kN} E-H_{kN})=0, \label{eq:DN}
\end{equation}
which is a polynomial of degree $kN$ with respect to $E$ for a $k$-band non-Hermitian Hamiltonian under OBC. Generally, such polynomials satisfy the following recurrence relation~\cite{zakrajvsek2004pascal,bacher2002determinants,supp}:
\begin{equation}
    D_{kN}(E)=\sum\limits_{m=1}^{s}C_m(\{t_{\alpha\beta}\},E)D_{k(N-m)}(E), \label{eq:recursion}
\end{equation}
where $s=\binom{l+r}{l}=\frac{(l+r)!}{l!r!}$ is the order of the recurrence relation, and $l$ ($r$) represents the range of hopping to the left (right). The coefficients $C_m(\{t_{\alpha\beta}\},E)$ are homogeneous functions of degree $km$, i.e., $C_m(\{xt_{\alpha\beta}\},xE)=x^{km}C_m(\{t_{\alpha\beta}\},E)$. 

Associated with the recurrence relation in Eq.~\eqref{eq:recursion}, we can define the following characteristic equation: 
\begin{equation}
    z^s=\sum\limits_{m=1}^{s}C_m(\{t_{\alpha\beta}\},E)z^{s-m}, \label{eq:char_eq_of_recurrence}
\end{equation}
whose solutions, $z_i(E), i=1,2,\cdots,s$ re-express the characteristic polynomial in Eq.~\eqref{eq:DN} on the spectrum as: 
\begin{equation}
    D_{kN}(E)=\sum\limits_{i=1}^sd_i(E)z_i(E)^N, \label{eq:DN_z}
\end{equation}
where $d_i(E)$ are determined by the initial conditions $D_0(E)=1,D_k(E),\cdots,D_{(s-1)k}(E)$ of the recurrence relation in Eq.~\eqref{eq:recursion}, which are the characteristic polynomials in Eq.~\eqref{eq:DN} corresponding to Hamiltonian matrices of size $0, k, 2k, \cdots, (s-1)k$, respectively. The characteristic polynomial of the null matrix (size zero) is conventionally defined as $1$, while the others can be computed directly through the standard determinant calculations in general, which are relatively straightforward due to their (initially) limited sizes. The OBC spectrum in the thermodynamic limit $N\rightarrow\infty$ is comprised of values of $E$ satisfying either of the following two conditions~\cite{Beraha1975,supp}: 
\begin{enumerate}[(i).]
    \item The two solutions with maximal absolute values have the same moduli
    $|z_i(E)|=|z_j(E)|\geq |z_{m\neq i,j}(E)|$;
    \item $d_i(E)=0$ in Eq.~\eqref{eq:DN_z} for the solution with the maximal absolute value $|z_i(E)|>|z_{j\neq i}(E)|$.
\end{enumerate} 

Note that condition (i) contains an infinite number of $E$ in a continuous spectrum of solutions and a similar form with the GBZ condition $|\beta_p(E)|=|\beta_{p+1}(E)|$ in the non-Bloch band theory, thus corresponding to the continuous bulk OBC spectra of the non-Hermitian systems. For example, the recurrence relation for the Hatano-Nelson model:
\begin{equation}
H_{\text{HN}}=\sum_i[(t+\gamma)c_{i+1}^\dagger c_i+(t-\gamma)c_{i}^\dagger c_{i+1}],    
\end{equation}
is:
\begin{equation}
    D_N^{\text{HN}}(E)=ED_{N-1}^{\text{HN}}(E)-(t^2-\gamma^2)D_{N-2}^{\text{HN}}(E),
\end{equation}
whose characteristic equation:
\begin{equation}
    z^2=Ez-(t^2-\gamma^2),
\end{equation}
has a degree $s=2$. Condition (i) $|z_1|=|z_2|=\sqrt{t^2-\gamma^2}$ offers a continuous set of solutions and differs from the GBZ condition $|\beta_1|=|\beta_2|=\sqrt{(t+\gamma)/(t-\gamma)}$ only by a constant. As a result, the spectrum $E=2\sqrt{t^2-\gamma^2}\cos\theta$ is consistent with the non-Bloch band theory~\cite{Kohei2020}.

On the other hand, condition (ii),  combined with the initial conditions of Eq.~\eqref{eq:DN_z}, gives several polynomial equations about $E$ and thus a discrete set of solutions. These polynomial equations are generally quite different from the non-polynomial form of condition (i) and generally insensitive of the system size $N$. Hence, we expect condition (ii) to contain the edge spectra of non-Hermitian systems under OBCs, which is beyond the scope of the non-Bloch band theory and the GBZ with only information of the bulk spectra~\cite{Kazuki2019, Zhesen2020}. For example, the spectrum equation of the Hatano-Nelson model following condition (ii) and approach elaborated in subsequent sections is:
\begin{equation}
    t^2-\gamma^2=0,
\end{equation}
which yields zero solutions, thus edge modes, which is consistent with the absence of edge modes in the single-band Hatano-Nelson model and, more generally, condition (ii) for the edge spectra. 

The recurrence method is remarkably useful for non-Hermitian models, especially models with large unit cells ($k\gg 1$) that are exceedingly complex from a non-Bloch theory perspective. As we will show, even for a large unit cell, the recurrence method gives an analytic expression for the OBC spectra of 1D non-Hermitian nearest-neighbor tight-binding models in the thermodynamic limit $N\rightarrow\infty$. Computationally, on the other hand, the recurrence method transfers the OBC spectrum problem to the problem of solving polynomials of degree $k$, which reduces the time complexity to $O(Nk^2)$ with the Durand-Kerner algorithm~\cite{kerner1966,supp}. 

Aside from the bulk band, edge states may emerge as isolated levels in the spectrum under OBC, intimately related to the topological invariants of the non-Hermitian systems~\cite{Yao2018,Yao2018Chern}. Based upon condition (ii), the recurrence method provides targeted computations of the edge spectra in general non-Hermitian models. By restricting the computation to the isolated solutions of $d_i(E)=0$, the recurrence method is much more efficient than the previous method based upon the edge matrix~\cite{Xiong2024}. It also provides an effective means to study edge physics~\cite{Li2023scale,Nakamura2023,Schindler2023,Landi2022}, i.e., how local perturbations at the OBC edges impact the edge spectrum, e.g., when establishing the stability of the topological edge modes, by tracking the changes to the $d_i(E)$ from the recurrence's initial conditions. 

\section{Analytical expression of OBC bulk spectrum} \label{sec:k_periodic}

\subsection{Periodic nearest-neighbor-hopping non-Hermitian models}

As an explicit demonstration of the recurrence method, we start with the non-Hermitian models with $k$-periodic onsite potential and nearest-neighbor (NN) hopping under OBCs: 
\begin{equation}
    H_1=\sum\limits_{l=1}^{N-1}(t_{l}c_{l+1}^\dagger c_l+t_{l}'c_l^\dagger c_{l+1}+V_ln_l), \label{eq:periodic_Hatano_Nelson}
\end{equation}
where the NN hoppings $t_{l}=t_{l+k}$, $t_{l}'=t_{l+k}'$ and the onsite potentials $V_l=V_{l+k}$ both have a real-space period of $k$. The corresponding recurrence relation is:
\begin{equation}
    D_N=(E-V_N)D_{N-1}-t_{N-1}t_{N-1}'D_{N-2}.
\end{equation}

Note that in order to implement the recurrence method, the coefficients of the recurrence relation in Eq.~\eqref{eq:recursion} should be independent of $N$. To address this issue, we need to derive the following size-independent recurrence relation by using the transfer matrix method~\cite{supp}:
\begin{equation}
    D_{kN}=p_k(E)D_{k(N-1)}-tD_{k(N-2)}, \label{eq:recurrence_method}
\end{equation}
where $t\equiv\prod_{i=1}^kt_{i}t_{i}'$ and $p_k(E)$ is a polynomial of $E$ with degree $k$ \cite{supp}:
\begin{equation}
    p_k(E)=\Delta_{1,k}-t_{k}t_{k}'\Delta_{2,k-1}, \label{eq:pk}
\end{equation}
where $\Delta_{i,j}$ for $1\leq i\leq j \leq k$ in Eq. \eqref{eq:pk} are given by
\begin{equation}
    \Delta_{i,j}=\det\begin{pmatrix}
        E-V_i & 1 & &  \\
        t_{i}t_{i}' & E-V_{i+1} & 1 &  \\
        & \ddots & \ddots & 1  \\
        & &t_{j-1}t_{j-1}' & E-V_j \\
    \end{pmatrix}, \label{eq:Delta}
\end{equation}
and $\Delta_{i,i-1}=1$, otherwise $\Delta_{i,j}=0$. 

Since $\Delta$ is a sparse matrix and satisfies $\Delta_{i,j}=(E-V_j)\Delta_{i,j-1}-t_{j-1}t_{j-1}'\Delta_{i,j-2}$, we can derive the coefficients of $p_k(E)$ recursively when the lattice period $k$ is large \cite{supp}. The recurrence relation in Eq.~\eqref{eq:recurrence_method} is characterized by the equation:
\begin{equation}
    z^2=p_k(E)z-t, \label{eq:characteristic_eq_recurrence_relation}
\end{equation}
with solutions $z_{\pm}=[p_k(E)\pm\sqrt{p_k(E)^2-4t}]/2$. From Eq.~\eqref{eq:DN_z}, the determinant $D_{kN}$ is:
\begin{equation}
    D_{kN}=d_+(E)z_+(E)^N+d_-(E)z_-(E)^N. \label{eq:DkN}
\end{equation}

Using the condition (i) $|z_+|=|z_-|$ in the recurrence method, we find that the OBC bulk spectra of the periodic NN hopping non-Hermitian models in Eq.~\eqref{eq:periodic_Hatano_Nelson} are given by:
\begin{equation}
    p_k(E)=2\sqrt{t}\cos\theta, \quad  \theta\in[0,\pi], \label{eq:pk_tcos}
\end{equation}
which is an analytical expression for the OBC spectra of multiple-band Hamiltonians in Eq.~\eqref{eq:periodic_Hatano_Nelson} through the recurrence method. Therefore, our formalism transforms the OBC spectra problems of Eq.~\eqref{eq:periodic_Hatano_Nelson} to a sequence of polynomial equations in Eq.~\eqref{eq:pk_tcos}. For a lattice period or sublattice size $k$ that is arbitrarily large, we can solve Eq.~\eqref{eq:pk_tcos} numerically through an iterative method called the Durand-Kerner method~\cite{kerner1966,supp} with manageable cost $\sim O(Nk^2)$, advantageous over the non-Bloch band theory. 

More specifically, consider the non-Hermitian Su-Schrieffer-Heeger (SSH) model with $k=2$ bands~\cite{Yao2018,Lieu2018}:  
\begin{align}
    H_{\text{SSH}}=&\sum\limits_i[(u_1+\gamma/2)c_{iA}^\dagger c_{iB}+(u_1-\gamma/2)c_{iB}^\dagger c_{iA}\\
    &+u_2c_{i+1,A}^\dagger c_{iB}+u_2c_{iB}^\dagger c_{i+1,A}], \label{eq:nH_ssh}
\end{align}
which is a special case of Eq.~\eqref{eq:periodic_Hatano_Nelson} with $t_{1}=u_1-\gamma/2$, $t_{1}'=u_1+\gamma/2$, $t_{2}=t_{2}'=u_2$, and $V_1=V_2=0$. 

According to the recurrence method, we obtain from Eq.~\eqref{eq:pk} and Eq.~\eqref{eq:Delta}: 
\begin{equation}
    p_{\text{SSH}}(E)=E^2-u_1^2-u_2^2+\frac{\gamma^2}{4}, t_{\text{SSH}}=u_1^2u_2^2-\frac{\gamma^2 u_2^2}{4}.  \label{eq:p_ssh_t_ssh}
\end{equation}

Consequently, the bulk spectrum of the non-Hermitian SSH model under OBC is analytically obtainable from Eq.~\eqref{eq:pk_tcos}: 
\begin{equation}
    E_{\text{SSH}}=\pm\sqrt{u_1^2+u_2^2-\gamma^2/4+2\sqrt{(u_1^2-\gamma^2/4)u_2^2}\cos\theta}. \label{eq:nH_ssh_RM} 
\end{equation}
 
For comparison, the non-Bloch band theory replaces the momentum $k$ with the non-Bloch parameter $\beta \equiv e^{ik}$ and solves the eigenvalue equation $\det[H(\beta)-E]=0$ instead~\cite{Kazuki2019}. Subsequently, its roots $\beta_1(E)$ and $\beta_2(E)$ yield the GBZ conditions $|\beta_1|=|\beta_2|=\sqrt{|u_1-\gamma/2|/|u_1+\gamma/2|}$, and the corresponding OBC spectrum takes the form~\cite{Yao2018}:
\begin{align}
    E^2&=u_1^2+u_2^2-\gamma^2/4+u_2\sqrt{|u_1^2-\gamma^2/4|}[\text{sgn}(u_1+\gamma/2)e^{i\theta} \\
    &+\text{sgn}(u_1-\gamma/2)e^{-i\theta}
    ], \quad \theta\in[0,2\pi], \label{eq:nH_ssh_GBZ} 
\end{align}
indeed, fully consistent with the results in Eq.~\eqref{eq:nH_ssh_RM} from our recurrence method. 

\subsection{Non-Hermitian Rice-Mele model} \label{subsec:nH_RM}

As another example of the applicability of the recurrence method, we consider the non-Hermitian Rice-Mele model~\cite{Rice1982, Wang2018}: 
\begin{align}
    H_{\text{nRM}}&=\sum\limits_i(w_1b_i^\dagger a_{i+1}+w_2a_{i+1}^\dagger b_i)+\sum\limits_i(v_1a_{i}^\dagger b_i+v_2b_i^\dagger a_i) \\
    &+\sum\limits_i V(a_i^\dagger a_i-b_i^\dagger b_i), \label{eq:nH_Rice_Mele}
\end{align}
where $w_{1/2}$ denote the intercell hoppings between the adjacent unit cells, $v_{1/2}$ denote the intracell hoppings between sites A and B in the same unit cell, and $V$ is a staggered onsite potential, as shown in Fig.~\ref{fig:Rice_Mele}(a). The $V=0$ case recovers the non-Hermitian SSH model in Eq.~\eqref{eq:nH_ssh} with vanishing onsite potentials. The asymmetric intracell/intercell hoppings, e.g., $v_1=v-\gamma$ versus $v_2=v+\gamma$ with a finite $\gamma$, introduce non-Hermiticity into the model and can be realized in experiments by reservoir engineering in the ultra-cold atom systems and imaginary gauge field in coupled optical micro-ring resonators~\cite{Gong2018, Liu2020Mott, Longhi2015}. 

Following the recurrence method beginning in Eq.~\eqref{eq:recurrence_method}, the polynomial $p_{\text{nRM}}(E)$ and $t_{\text{nRM}}$ in Eq.~\eqref{eq:nH_Rice_Mele} for the non-Hermitian Rice-Mele model are: 
\begin{subequations}
\begin{eqnarray}
    p_{\text{nRM}}(E)&=&(E-V)(E+V)-v_1v_2-w_1w_2, \\
    t_{\text{nRM}}&=&v_1v_2w_1w_2.
\end{eqnarray}
\end{subequations}

Thus, following Eq.~\eqref{eq:pk_tcos}, the bulk OBC spectrum of the non-Hermitian Rice-Mele model \eqref{eq:nH_Rice_Mele} is: 
\begin{equation}
E_{\text{nRM}}=\pm\sqrt{V^2+v_1v_2+w_1w_2+2\sqrt{v_1v_2w_1w_2}\cos\theta}, \label{eq:RM_bulk_OBC}
\end{equation}
which retains the symmetric pattern of the Hermitian Rice-Mele model spectrum and is consistent with the results of the non-Bloch band theory~\cite{Mandal2024}. 

In particular, the OBC spectrum in Eq. \eqref{eq:RM_bulk_OBC} becomes real when and only when:  
\begin{subequations}
\begin{eqnarray}
    V^2+v_1v_2+w_1w_2&>&2\sqrt{v_1v_2w_1w_2}, \\
    v_1v_2w_1w_2&>&0,
\end{eqnarray}
\end{subequations}
which are satisfied in two scenarios: first, both $v_1v_2, w_1w_2>0$ for an arbitrary onsite potential $V$, as shown in Fig.~\ref{fig:Rice_Mele}(b-c); the other is both $v_1v_2, w_1w_2<0$ and $|V|>\sqrt{-v_1v_2}+\sqrt{-w_1w_2}$. The latter case introduces possibilities of the non-Bloch PT symmetry breaking~\cite{Longhi2019Opt,Nie2024,Xiao2021,Hu2024} as the onsite potential $V$ crosses and falls between the transition points:
\begin{equation}
    V_c=\pm(\sqrt{-v_1v_2}+\sqrt{-w_1w_2}). 
\end{equation}
For instance, $V_c=\pm1.4$ for parameters $v_1=-0.8$, $v_2=0.2$, $w_1=-1$, and $w_2=1$ so that $v_1v_2, w_1w_2<0$, and the OBC spectrum of the non-Hermitian Rice-Mele model in Eq.~\eqref{eq:nH_Rice_Mele} remains real at $V=1.5$ [Fig.~\ref{fig:Rice_Mele}(g)] yet undergoes an non-Bloch PT transition at $V=1.4$ [Fig.~\ref{fig:Rice_Mele}(f)] and develops complex values afterwards [Fig.~\ref{fig:Rice_Mele}(d-e)]. 

\begin{figure}[!ht]
    \centering
    \includegraphics[width=0.98\linewidth]{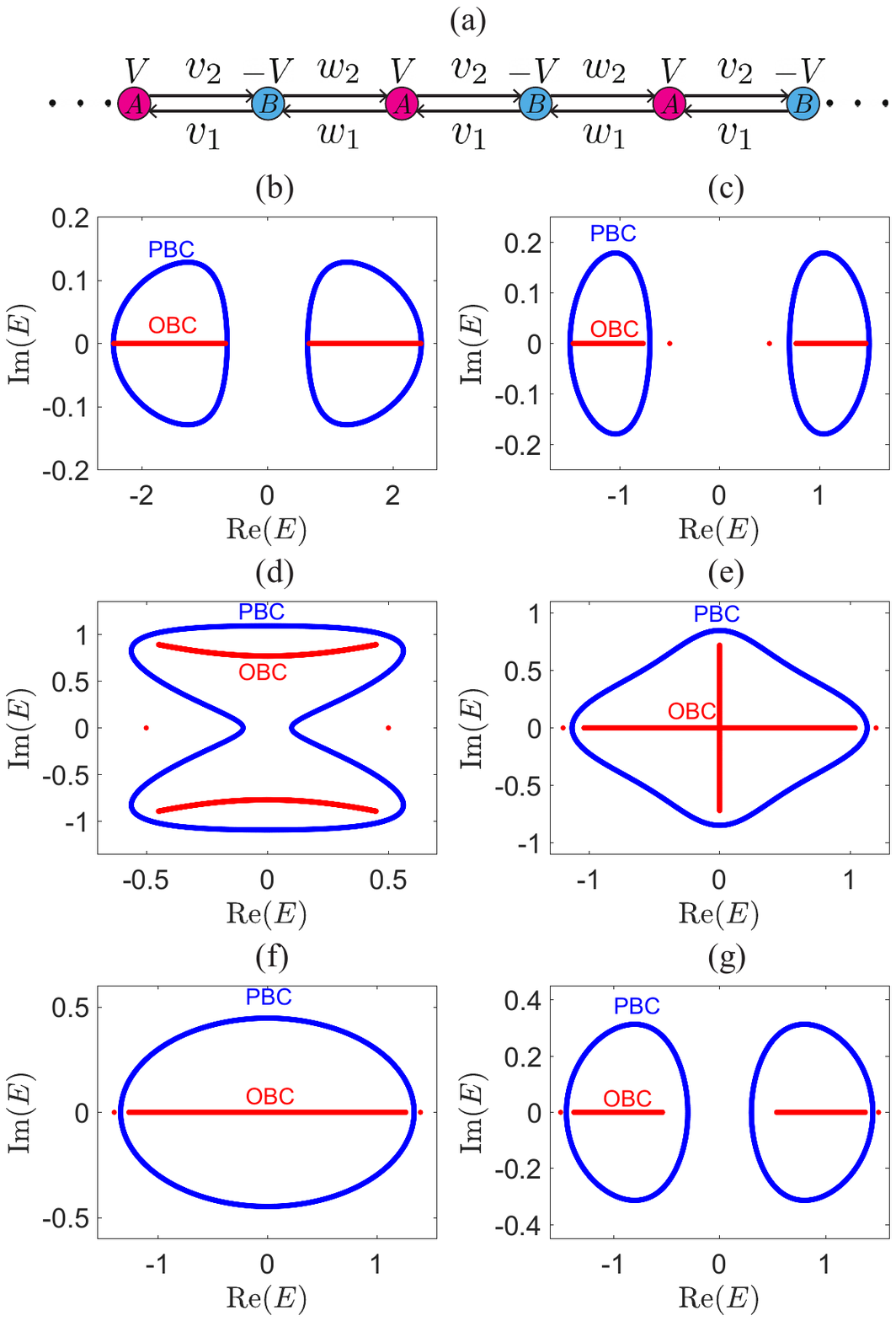}
    \caption{We compare the OBC (red curves and dots) and PBC (blue curves) spectra of the non-Hermitian Rice-Mele model in Eq.~\eqref{eq:nH_Rice_Mele} with two sublattices A and B and asymmetric hoppings $v_1\neq v_2^*$ or $w_1\neq w_2^*$, illustrated in (a). The OBC spectra are fully real when both $v_1v_2, w_1w_2>0$, irrespective of $V$: (b) $v_1=2$, $v_2=1$, $w_1=1.2$, $w_2=0.8$, and $V=0.5$; (c) $v_1=0.8$, $v_2=0.2$, $w_1=1.2$, $w_2=0.8$, and $V=0.5$; (d) $v_1=0.8$, $v_2=0.2$, $w_1=-1$, $w_2=1$, and $V=0.5$. However, when both $v_1v_2, w_1w_2<0$, a non-Bloch PT symmetry breaking may occur: (e-g) $v_1=-0.8$, $v_2=0.2$, $w_1=-1$, and $w_2=1$, thus the transition points are at $V_c=1.4$. The spectrum is complex for (d) and (e) $V=1.2$ at a critical point at (f) $V=1.4$ and becomes fully real for (g) $V=1.5$. Note that $|v_1v_2|>|w_1w_2|$ and thus no edge state exists in (b), while $|v_1v_2|<|w_1w_2|$ and two isolated edge states appear at $E=\pm V$ (red dots) in (c-g). }
    \label{fig:Rice_Mele}
\end{figure}

\section{OBC edge spectrum: effects of system size and perturbations} \label{sec:edge_spectra}

\subsection{Edge spectra by recurrence method} \label{sec:subedge_spectra}

The recurrence method also provides an effective and targeted approach for the edge part of the OBC spectra in non-Hermitian systems. For example, for the edge spectrum of the nearest-neighbor-hopping model in Eq.~\eqref{eq:periodic_Hatano_Nelson}, we consider the following initial conditions of Eq.~\eqref{eq:DkN}: 
\begin{eqnarray}
    D_0(E)&=&d_+(E)+d_-(E)=1, \\\nonumber
    D_k(E)&=&d_+(E)z_+(E)+d_-(E)z_-(E)=\Delta_{1,k}, \label{eq:D0_Dk}
\end{eqnarray} 
where $\Delta_{1,k}$ is given in Eq. \eqref{eq:Delta}. 

Using condition (ii) that $d_i(E)=0$ for $|z_i|>|z_{j\neq i}|$ in Sec.~\ref{sec:recurrence_method} along with Eq.~\eqref{eq:characteristic_eq_recurrence_relation}, we find that the OBC edge spectrum is within~\cite{supp}:
\begin{equation}
    D_k(E)^2-D_k(E)p_k(E)+t=0, \label{eq:Dk_pk_t}
\end{equation}
along with either of the following scenarios: 
\begin{eqnarray}
    D_k(E)&=&z_+(E), \quad |z_+(E)|<|z_-(E)|; \nonumber\\
    D_k(E)&=&z_-(E), \quad |z_-(E)|<|z_+(E)|. \label{eq:edge_OBC_condition}
\end{eqnarray}

Previously, the edge spectra of non-Hermitian systems under OBCs are obtainable via the edge matrix~\cite{Xiong2024}. Consider the $m$-band non-Hermitian Hamiltonian $H(\beta)=\sum_{i=-r}^l h_i\beta^i$ after the substitution $e^{ik}\rightarrow\beta$, the block element of the $mr\times mr$ edge matrix $M_{\text{edge}}$ is defined as~\cite{Xiong2024}: 
\begin{equation}
    [M_{\text{edge}}(E)]_{ab}=\frac{1}{2\pi i}\oint_{C}\beta^{a-b}[H(\beta)-E\mathbbm{1}_m]^{-1}\frac{d\beta}{\beta},
\end{equation}
where $a,b=1,2,\cdots,r$, and the counterclockwise integral contour encloses the origin and the first $p$ solutions $\beta_1(E), \cdots, \beta_p(E)$ of $\det(H(\beta)-E)=0$. The isolated edge modes are determined by the relation $\det M_{\text{edge}}=0$~\cite{delvaux2012}. Compared to the edge matrix, the recurrence method is more efficient, as it restricts the edge-spectrum problem to the few zero-coefficient solutions of $d_i(E)=0$, while through the edge matrix, $\det M_{\text{edge}}(E)\neq 0$ - a complicated and costly matrix equation, especially when $mr$ is large - has to be verified for all $E$ outside the OBC spectra through the edge matrix~\cite{Xiong2024}. 

For example, given the non-Hermitian SSH model in Eq.~\eqref{eq:nH_ssh}, $D_2(E)=E^2-u_1^2+\gamma^2/4$, and Eq.~\eqref{eq:Dk_pk_t} simply becomes $E^2u_2^2=0$. Therefore, the recurrence method directly suggests that the edge spectra of non-Hermitian SSH models are zero modes, which emerge in the topological non-trivial regime $|u_1^2-\gamma^2/4|<u_2^2$~\cite{Xiong2018,Yao2018}, following the conditions in Eq.~\eqref{eq:edge_OBC_condition}. For comparison, the edge matrix of the non-Hermitian SSH model has a zero determinant at $E=0$ when and only when $|u_1^2-\gamma^2/4|<u_2^2$~\cite{Xiong2024}, in full alignment with our results. However, the knowledge $E=0$ has to be either presumed or concluded from various trial $E$ values. 

Likewise, the initial conditions $D_0$, $D_2$ of the non-Hermitian Rice-Mele models in Eq.~\eqref{eq:nH_Rice_Mele} are given by: 
\begin{equation}
    D_2^{\text{nRM}}(E)=E^2-V^2-v_1v_2, \label{eq:Rice_Mele_D2}
\end{equation}
according to our recurrence method. Eq.~\eqref{eq:Dk_pk_t} describing the the OBC edge spectrum now becomes: 
\begin{equation}
E^2w_1w_2-V^2w_1w_2=0,
\end{equation}
therefore $E=\pm V$. Combined with the conditions in Eq.~\eqref{eq:edge_OBC_condition}, such edge modes exist under and only under the condition $|v_1v_2|<|w_1w_2|$, as shown in Fig.~\ref{fig:Rice_Mele}(c-g), where two isolated edge modes appear aside from the continuous bulk OBC spectra. 

\subsection{Non-Hermitian size-parity effect}

Until now, we have assumed that the overall system size is an integer multiple of the number of bands, i.e., the size of the unit cells and the number of sublattices, e.g., an overall even number of sites for the non-Hermitian Rice-Mele and SSH models with period $k=2$. However, overall system size can play a significant role in topological phases and transitions in non-Hermitian systems~\cite{Li2020, Liu2020helical}. Here, we study the variation of the OBC spectra versus the parity of system sizes, sometimes dubbed the non-Hermitian size-parity effect, via the recurrence method. 

In the momentum space, the Hermitian Rice-Mele model under PBC takes the following form:
\begin{equation}
    H_{\text{RM}}(k)=h_x(k)\sigma_x+h_y(k)\sigma_y+h_z(k)\sigma_z, \label{eq:Hermitian_Rice_Mele}
\end{equation}
where $h_x=v+w\cos k$, $h_y=w\sin k$, and $h_z=V$. It has time-reversal symmetry $\mathcal{T}H_{\text{RM}}(k)\mathcal{T}^{-1}=H_{\text{RM}}(k)$, a generalized chiral symmetry $\hat{\Gamma}^{PT}H_{\text{RM}}(k)(\hat{\Gamma}^{PT})^{-1}=-H_{\text{RM}}(k)$, and a generalized particle-hole symmetry $\hat{\mathcal{C}}^{PT}H_{\text{RM}}(k)(\hat{\mathcal{C}}^{PT})^{-1}=-H_{\text{RM}}(k)$, where $\hat{\Gamma}^{PT}=i\sigma_y\hat{\mathcal{K}}$, $\hat{\mathcal{C}}^{PT}=i\sigma_y\otimes\hat{\mathcal{I}}_k$, $\hat{\mathcal{I}}_k\equiv (k\rightarrow -k)$ and 
$\hat{\mathcal{K}}$ being the complex conjugation~\cite{Han2020,Qi2023}. These symmetries ensure the model's symmetric bands $\pm E_k$. 

Due to the broken chiral (sublattice) symmetry, two separated nonzero edge modes emerge in the gap between the two symmetric bands~\cite{Wang2018, Han2020, Lin2020}. These two edge modes are localized on the opposing edges and evolve with model parameters, which can implement bidirectional edge transmission~\cite{Qi2023}. In the thermodynamic limit $L\rightarrow\infty$, the edge mode with $E=V$ is left localized at the left open boundary~\cite{Lin2020}: 
\begin{equation}
    \psi_{i, \tau}\propto \delta_{\tau, A}(-1)^{i+1}e^{-\kappa i}, \quad \kappa=\ln\frac{w}{v}, \label{eq:RM_edge_wave_function}
\end{equation}
where $i$ labels the unit cell and $\tau=A, B$ denotes the sublattices. $\delta_{\tau, A}$ indicates that the edge mode weighs exclusively on the $A$ sublattice. It is apparent from the decay constant $\kappa$ that such a localized edge mode exists only when $v<w$. Likewise, when and only when $v>w$, a right-localized edge mode with $E=-V$ emerges on the $B$ sublattices at the right open boundary.

\begin{figure*}[htb]
    \centering
    \includegraphics[width=\linewidth]{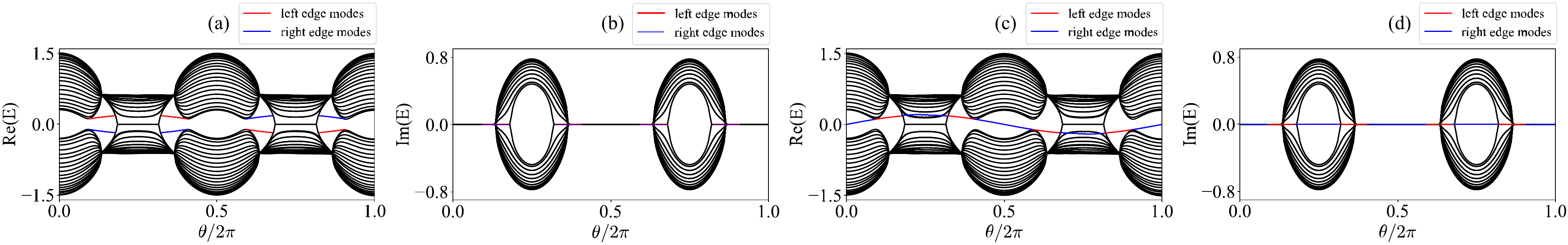}
    \caption{The real [(a) and (c)] and imaginary [(b) and (d)] parts of the OBC spectra for non-Hermitian Rice-Mele models in Eq. \eqref{eq:nH_Rice_Mele} on even-size [$2N=40$ in (a) and (b)] and odd-size [$2N+1=41$ in (c) and (d)] systems display a size-parity effect. The black lines represent the OBC bulk spectrum, and the red (blue) lines represent the left (right) localized edge modes. The model parameters $v_1=0.9+1.2\sin\theta$, $v_2=0.9-1.2\sin\theta$, $w_1=w_2=0.6$, and $V=0.2\sin\theta$ are variable with respect to $\theta$. For the $2N$-size models, a pair of edge modes at $E=\pm V$ appear when $0.56<\sin\theta<0.90$, consistent with condition $|v_1v_2|<|w_1w_2|$; for the $(2N+1)$-size models, an isolated edge mode at $E=V$ always exists irrespective of $\theta$, yet switches between the right and left edges as $\theta$ evolves. }
    \label{fig:RM_odd_even}
\end{figure*}

It is important to note that the left (right)-localized edge mode has the same energy $E$ as the onsite potential $V$ ($-V$) at the left (right) boundary. Intuitively, as the onsite potential is $V$ at both the left and right boundaries for a system with an odd number of sites, the corresponding edge modes should have $E=V$. Furthermore, Eq.~\eqref{eq:RM_edge_wave_function} determines the localization direction of isolated edge modes: if $v<w$ ($v>w$), edge modes can localize at the left (right) boundary. Therefore, we can manipulate the appearance and location of the edge modes through the model parameters. Such an interesting phenomenon has been widely used in topological pumping and quantum state transfer~\cite{Qi2023, Longhi2019topo, Chang2023}. 

Here, we demonstrate via the recurrence method that this size-parity effect can be extended to the non-Hermitian Rice-Mele model in Eq.~\eqref{eq:nH_Rice_Mele} and also in the absence of generalized chiral symmetry. In Sec.~\ref{sec:subedge_spectra}, we have shown that the existence condition of the edge modes is $|v_1v_2|<|w_1w_2|$, which returns to the $v<w$ condition in Hermitian cases. There, we have to choose the zeroth and second principle minors $D_0$ and $D_2$ as the initial conditions in an even-size non-Hermitian Rice-Mele model. However, in an odd-size non-Hermitian Rice-Mele model, we should change the initial conditions to the first and third principle minors $D_1$ and $D_3$ since the OBC spectrum is determined by $D_{2N+1}$. The expressions for $D_1$ and $D_3$ are: 
\begin{equation}
    D_1^{\text{nRM}}(E)=E-V, \quad D_3^{\text{nRM}}(E)=(E-V)p_{\text{nRM}}(E).
\end{equation}
Substitute them into:
\begin{equation}
    D_{2N+1}(E)=\sum\limits_{i=\pm}d_i(E)z_i(E)^N,
\end{equation}
we obtain: 
\begin{eqnarray} 
    E-V&=&d_+^{\text{nRM}}+d_-^{\text{nRM}}, \\ \nonumber
    (E-V)p_{\text{nRM}}&=&d_+^{\text{nRM}}z_+^{\text{nRM}}+d_-^{\text{nRM}}z_-^{\text{nRM}}.
\end{eqnarray}

Finally, from condition (ii) in Sec.~\ref{sec:recurrence_method} and $z_{\pm}^{\text{nRM}}=(p_{\text{nRM}}\pm\sqrt{p_{\text{nRM}}^2-4t_{\text{nRM}}})/2$, regardless of the bulk parameters $v_1$, $v_2$, $w_1$, and $w_2$, there always exists one and only one isolated edge mode with energy $E=V$ in the non-Hermitian Rice-Mele model with an odd number of sites, i.e., $2N+1$. We summarize example results in Fig.~\ref{fig:RM_odd_even}, where an isolated edge mode always exists yet switches between the left and right edges with an odd number of sites. Via the recurrence method, our results indicate that the size-parity effect in the Rice-Mele model, as well as the subsequent unidirectional topological edge excitation transmission and quantum state transfer~\cite{Qi2023, Longhi2019topo, Chang2023}, can be generalized to non-Hermitian scenarios and systems without generalized chiral symmetry.

\subsection{Robustness of topological edge modes}
\label{subsec:robust_edge_modes}

In the previous subsection, we studied the size-parity effect in non-Hermitian systems through the recurrence method as altered initial conditions. Another interesting scenario reflected in the modified initial conditions of the recurrence method is when boundaries encounter perturbations, which may impact the edge spectra of the systems~\cite{Li2023scale,Nakamura2023,Schindler2023,Landi2022}. Such problems are also helpful in evaluating the robustness of the edge modes as topological characterizations~\cite{Yao2018,Yao2018Chern}.  

For example, consider a non-Hermitian SSH model with an additional onsite potential on the left boundary:
\begin{equation}
    H_{\text{SSH}}'=H_{\text{SSH}}+V_1c_0^\dagger c_0-V_1c_1^\dagger c_1, \label{eq:perturbed_nH_ssh}
\end{equation}
where $H_{\text{SSH}}$ is in Eq.~\eqref{eq:nH_ssh}. Without loss of generality, we set $u_1=1$, $u_2=2$, $\gamma=1$, and an overall even number of sites. 

From the recurrence method, it is clear that the bulk spectrum of the boundary-perturbed model $H_{\text{SSH}}'$ remains unchanged while the edge spectrum is adjusted accordingly. The initial condition of the recurrence relation now becomes $D_2'(E)=E^2-V_1^2-u_1^2+\gamma^2/4$ and $D_0$ no longer equal to $1$:
\begin{equation}
    D_0'(E)=1-\frac{V_1(E-V_1)}{u_1^2-\gamma^2/4},
\end{equation}
for which we obtain $D_4'$ first and then use $D_4'=p_{\text{SSH}}D_2'-t_{\text{SSH}}D_0'$. Substituting $D_0'$ and $D_2'$ into Eq.~\eqref{eq:DkN}, and setting $d_\pm(E)=0$, we obtain the equation dictating the edge spectra of the perturbed non-Hermitian SSH model:  
\begin{equation}
    D_2'(E)^2-D_2'(E)D_0'(E)p_{\text{SSH}}(E)+t_{\text{SSH}}D_0'(E)^2=0,
\end{equation}
and, following condition (ii) in Sec.~\ref{sec:recurrence_method}, either of the two conditions: 
\begin{eqnarray}
    D_2'(E)&=&z_+(E)D_0'(E), \quad |z_+(E)|<|z_-(E)|, \nonumber \\
    D_2'(E)&=&z_-(E)D_0'(E), \quad |z_-(E)|<|z_+(E)|, \label{eq:edge_OBC_condition_perturbed}
\end{eqnarray} 
where $z_\pm=(p_{\text{SSH}}\pm\sqrt{p_{\text{SSH}}^2-4t_{\text{SSH}}})/2$ are identical to the original non-Hermitian SSH model.

\begin{figure}[ht]
    \centering
    \includegraphics[width=0.99\linewidth]{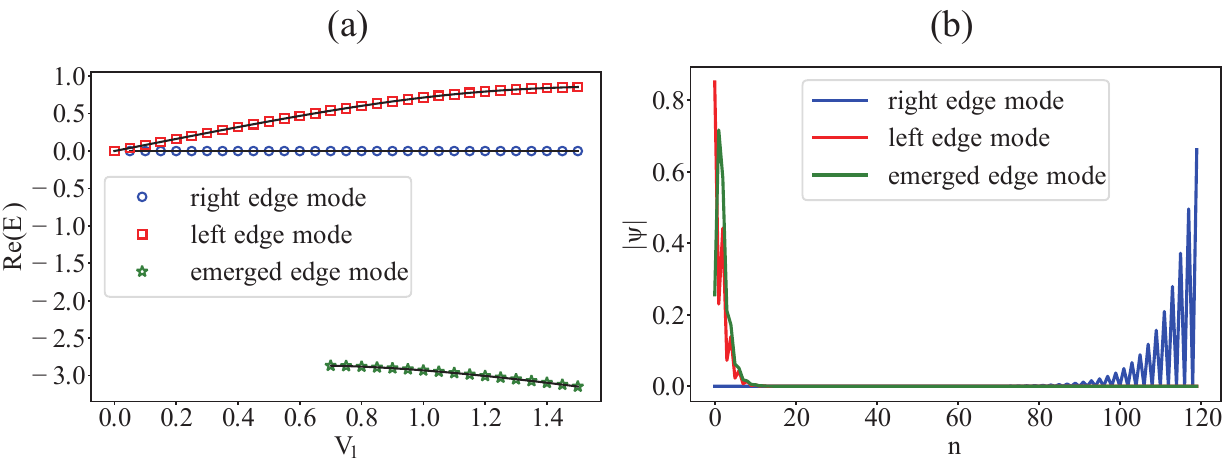}
    \caption{The spectra and localization properties of the edge modes for the boundary-perturbed non-Hermitian models in Eq.~\eqref{eq:perturbed_nH_ssh} show the stability of the topological edge modes. (a) As the perturbation $V_1$ varies on the left edge, the right edge mode (blue) stays unchanged, while the original left edge mode (red) develops an energy shift yet remains in the bulk gap. An additional edge mode emerges on the left boundary for $V_1>0.7$. (b) The wave functions of the respective edge modes show full localizations at $V_1=1.2$. $u_1=1$, $u_2=2$, and $\gamma=1$. }
    \label{fig:boundary_perturbation}
\end{figure}

We solve the equations numerically for the edge spectra of the perturbed non-Hermitian SSH models and the effect of the boundary perturbation $V_1$. The results are shown in Fig.~\ref{fig:boundary_perturbation}. As $V_1$ varies, the right edge mode keeps zero energy, as a gapped bulk separates it from the perturbation. On the contrary, the left edge mode develops an energy shift and splits from the right edge mode. The evolution of these edge modes under perturbations demonstrates the topological characteristics and stability of the non-Hermitian SSH models~\cite{Yao2018}. Interestingly, we observe an additional left edge mode emerging for $V_1>0.7$.

\section{Non-Hermitian Hofstadter model} \label{sec:nH_Hofstadter}

The energy spectrum of a non-interacting two-dimensional lattice model versus the presenting perpendicular magnetic field takes on a fractal and self-similar pattern commonly known as the Hofstadter butterfly~\cite{Hofstadter1976}, which has witnessed widespread appearances in various cutting-edge physics fields, including topological phases, quantum Hall effects, quasi-crystals, 2D materials~\cite{Hatsugai1993,dean2013hofstadter,ni2019observation,Lu2021hofstadter}. Recently, there have been studies on its interplay with non-Hermiticity~\cite{Shao2022, Chernodub2015, Sergey2014}, e.g., a biased quantum random walk model that forbids some of the hoppings while retaining the other hoppings of the electrons~\cite{Chernodub2015, Sergey2014}, whose spectrum in the presence of a commensurate magnetic flux can be solved exactly through the Chebyshev polynomials~\cite{Sergey2014}. Another non-Hermitian generalization of the Hofstadter butterfly considers a two-dimensional square lattice model with non-reciprocal hoppings in one of the directions and a perpendicular magnetic field~\cite{Shao2022}. Despite the non-reciprocity, the wave-packet trajectories still form closed orbits in the four-dimensional complex phase space in the complex semiclassical theory~\cite{Yang2024}; therefore, the quantization rule persists, and the Landau levels may take on still discrete yet complex energy levels~\cite{Shao2022, Yang2024}. On the other hand, non-Hermitian systems have a tendency and preference for real-valued OBC spectra.

\begin{figure}[htb]
    \centering
    \includegraphics[width=0.90\linewidth]{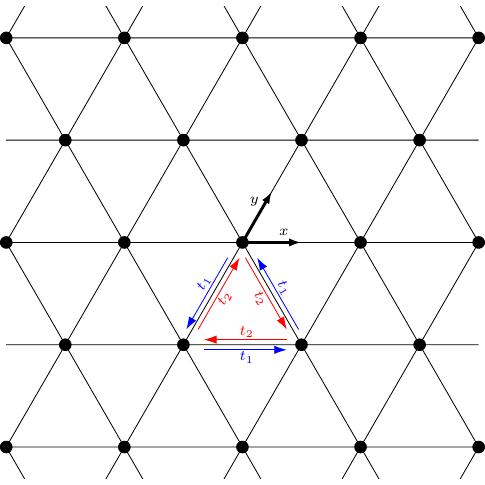}
    \caption{We consider the non-Hermitian Hofstadter butterfly on a two-dimensional triangular lattice. The model, given by the Hamiltonian in Eq. \eqref{eq:nH_triangle_lattice_model}, consists of unequal counterclockwise hoppings $t_1=\sqrt{(1-\delta)(1+\delta)}$ (blue arrows)  and clockwise hoppings $t_2=\sqrt{(1+\delta)(1-\delta)}$ (red arrows). Also, we introduce a perpendicular magnetic field with $p/2q$ magnetic flux quantum per triangle plaquette. We set $\boldsymbol{x}$ and $\boldsymbol{y}$ as our lattice vectors. }
    \label{fig:triangular_lattice}
\end{figure}

Here, as a concrete example of the interplay between the non-Hermiticity, the complex spectrum, and the magnetic field, we consider a non-reciprocal model on the two-dimensional triangle lattice, as shown in Fig.~\ref{fig:triangular_lattice}:
\begin{align}
    H_{\text{nHH}}=&\sum\limits_{x,y}(t_1c_{x+1,y}^\dagger c_{x,y}+t_2c_{x-1,y}^\dagger c_{x,y}+t_1c_{x-1,y+1}^\dagger c_{x,y} \\
    &+t_2c_{x+1,y-1}^\dagger c_{x,y}+t_1c_{x,y-1}^\dagger c_{x,y}+t_2c_{x,y+1}^\dagger c_{x,y}), \label{eq:nH_triangle_lattice_model}
\end{align}
where we introduce the non-reciprocity parameter $\delta$ as the difference between counterclockwise hoppings $t_1=\sqrt{(1-\delta)/(1+\delta)}$ and clockwise hoppings $t_2=\sqrt{(1+\delta)/(1-\delta)}$. We set $t_1 t_2=1$ as our unit of energy as well as $\hbar=e=1$ for simplicity. In the presence of a perpendicular magnetic field, each hopping acquires an additional phase $t'=te^{-i\theta}$, $\theta=\int \vec{A}\cdot d\vec{l}$ between the corresponding hopping's initial and final positions. Interestingly, the non-reciprocity in the clockwise and counterclockwise hoppings effectively acts as an imaginary magnetic field~\cite{Yang2024}, analogous to the imaginary gauge potential of the non-Hermitian Hatano-Nelson model~\cite{Hatano1996, Hatano1997, Longhi2015}. Hence, the non-reciprocal triangle lattice model in Eq.~\eqref{eq:nH_triangle_lattice_model} in the presence of a magnetic field $B_r$ is equivalent to a Hermitian triangle lattice model with unit hopping and a complex magnetic field $B=B_r+iB_i$, where the imaginary part $B_i$ relates to the non-reciprocity $\delta$: $e^{B_i}=(1+\delta)^3/(1-\delta)^3$~\cite{supp}.

\begin{figure*}[htb]
    \centering
    \includegraphics[width=0.78\linewidth]{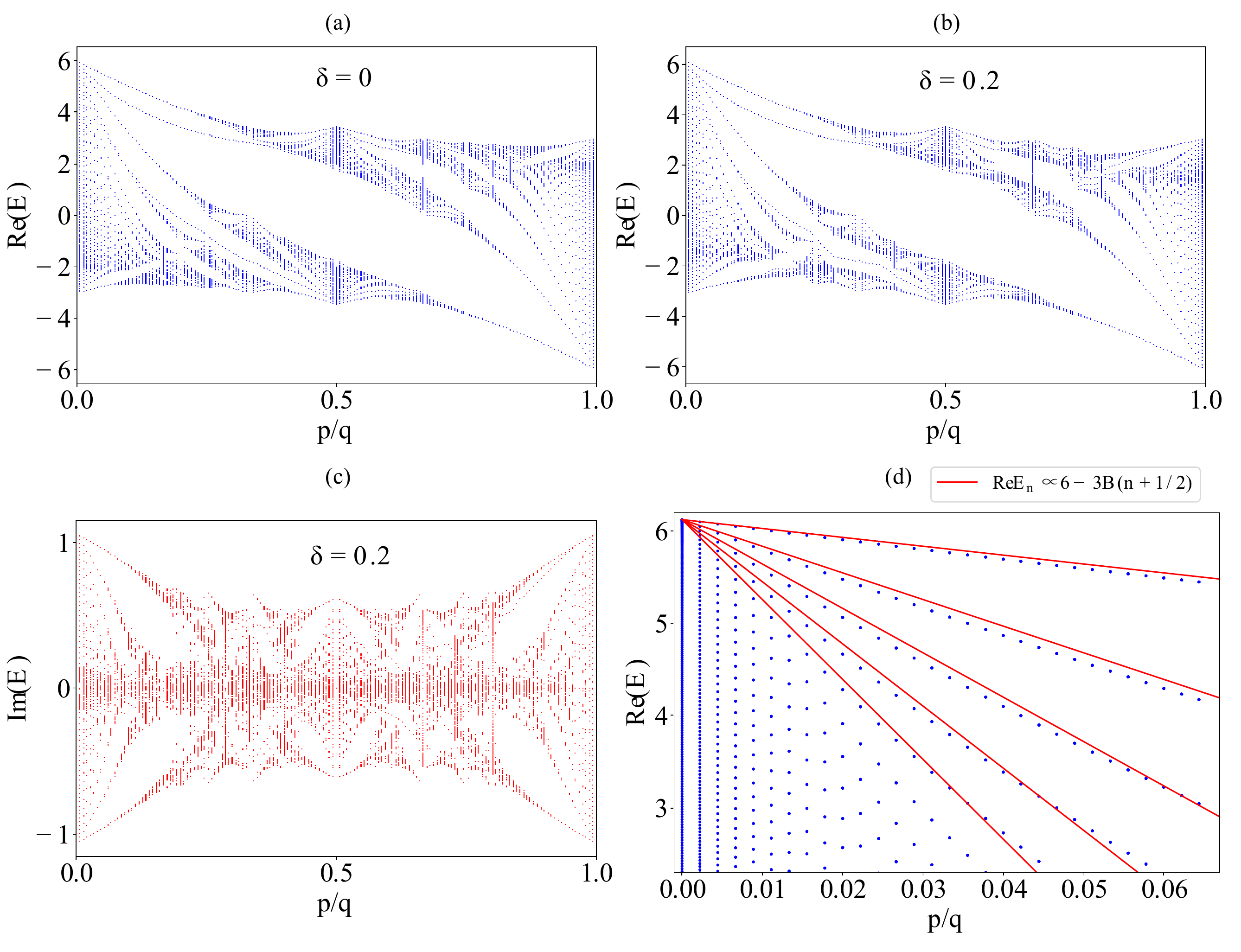}
    \caption{The spectra of the triangular lattice models in Eq.~\eqref{eq:nH_triangle_lattice_model} versus the applied perpendicular magnetic field ($p/q$ magnetic flux quantum per lattice plaquette) displays clear fractal and self-similar features resembling a (Hofstadter) butterfly, both in (a) Hermitian $\delta=0$ and (b-c) non-Hermitian $\delta=0.2$ scenarios. We consider a geometry with PBC in the $\boldsymbol{y}$ direction and OBC in the ${x}$ direction. The blue (b) and red (c) spectra are the real and imaginary parts, respectively. (d) The Landau fan structure near the band top (and bottom) exhibits a linear dependence on the magnetic field (red lines) consistent with $\sqrt{1-\delta^2}E_n\propto 6-3B(n+1/2), n=0, 1, 2, \cdots$ from the complex semiclassical theory~\cite{Yang2024,supp} and the interpretation of a complex effective magnetic field for a finite $\delta=0.2$. }
    \label{fig:nH_Hofstadter_fig1}
\end{figure*}

In addition, we consider the non-Hermitian Hofstadter butterfly with OBC in the $\boldsymbol{x}$ direction and PBC in the $\boldsymbol{y}$ direction, which, together with the choice of the Landau gauge $\vec{A}=(0, Bx)$, ensures $k_y$ as a good quantum number. For a particular $k_y$, we thus reduce the model to a one-dimensional non-Hermitian Hamiltonian with NN hopping: 
\begin{align}
    H_{\text{nHH}}(k_y)=&\sum\limits_{x}[t_{x}c_{x+1,k_y}^\dagger c_{x,k_y}+t_{x}'c_{x,k_y}^\dagger c_{x+1,k_y} \\
    &+V_xc_{x,k_y}^\dagger c_{x,k_y}], \label{eq:nH_complex_magnetic_field}
\end{align}
where the NN hoppings $t_{x}$, $t_x'$, and the onsite potential $V_x$ are in the Supplemental Materials \cite{supp} and generally complex-valued; therefore, the comparability with and similarity transformation to a Hermitian model no longer holds in such triangular-lattice Hofstadter butterfly models. Also, we consider a commensurate perpendicular magnetic field $B_r$, i.e., $B_r=2\pi (p/q)$ with coprime integers $p$ and $q$. Consequently, the Hamiltonian in Eq.~\eqref{eq:nH_complex_magnetic_field} is $q$-periodic ($k=q$). However, numerical diagonalization of a non-Hermitian Hamiltonian is highly susceptible to accumulated errors and sensitive to numerical precisions, especially in large systems~\cite{Zhesen2020, Colbrook2019, Yuto2020}. The non-Bloch band theory and GBZ method also face substantial obstacles and challenges for relatively large $q$, as the number of sublattices in the magnetic unit cell significantly elevates the computational cost. Hence, trustworthy calculations of energy spectra of Eq.~\eqref{eq:nH_complex_magnetic_field} for large $N$ and $q$ are exceedingly time-consuming~\cite{Zhesen2020}. Through the recurrence method, on the other hand, we transform the eigenvalue problem of Eq.~\eqref{eq:nH_complex_magnetic_field} to the root solutions of a series of $q$-th degree polynomials, which share the same coefficients except for the constant terms, as we have shown in Sec.~\ref{sec:k_periodic}. Numerically, we can solve these polynomials with an iterative algorithm called the Durand-Kerner method \cite{kerner1966,supp}, efficient and useful for simultaneously determining all complex roots of high-degree polynomials with a quadratic convergence rate. Therefore, the recurrence method avoids the numerical accuracy and instability issue in diagonalizing large non-Hermitian matrices and substantially reduces the computational complexity to $O(N)$ as long as $q\ll N$.

We summarize our results on the Hofstadter butterfly on the triangular lattice in Fig. \ref{fig:nH_Hofstadter_fig1} for $k_y=0$. Indeed, in the presence of a perpendicular magnetic field $B$, both the Hermitian ($\delta=0$) and non-Hermitian ($\delta=0.2$) models display the characteristic fractal and self-similar features intrinsic in a Hofstadter butterfly; see Fig.~\ref{fig:nH_Hofstadter_fig1}(a). On the other hand, the introduction of non-reciprocity $\delta$ deforms the spectra [Fig.~\ref{fig:nH_Hofstadter_fig1}(b)] and introduces imaginary parts [Fig.~\ref{fig:nH_Hofstadter_fig1}(c)]. Interestingly, the Landau fan structure is still present in the non-Hermitian case, as shown in Fig.~\ref{fig:nH_Hofstadter_fig1}(d), consistently and quantitatively validating the complex semiclassical theory in non-Hermitian systems~\cite{Yang2024,supp}. Notably, the magnetic unit cell at large $q$, e.g., a small $B$, becomes very large, and the computational cost and numerical stability quickly get out of control. However, it does not possess an apparent obstacle within the framework of our recurrence method, and we approach $q$ as high as $q_{\text{max}}=150$ as presented in Fig.~\ref{fig:nH_Hofstadter_fig1}. 

As previously stated, the non-Hermiticity $\delta$ introduces an imaginary magnetic field $iB_i$, so the overall effective magnetic field $B_c= B_r+iB_i$ is complex. Although $B_r$ needs to be commensurate and thus discrete, $\delta$ and $B_i$ are continuously variable, offering much finer and broader tunability to $B_c$. For instance, we can study a complex magnetic field $B_c=|B_c|\exp(i\theta)$ with a constant modulus $|B_c|=6$ that would have been incommensurate otherwise; see further details and results in the Supplemental Materials \cite{supp}.

\begin{figure*}[ht]
    \centering
    \includegraphics[width=0.98\linewidth]{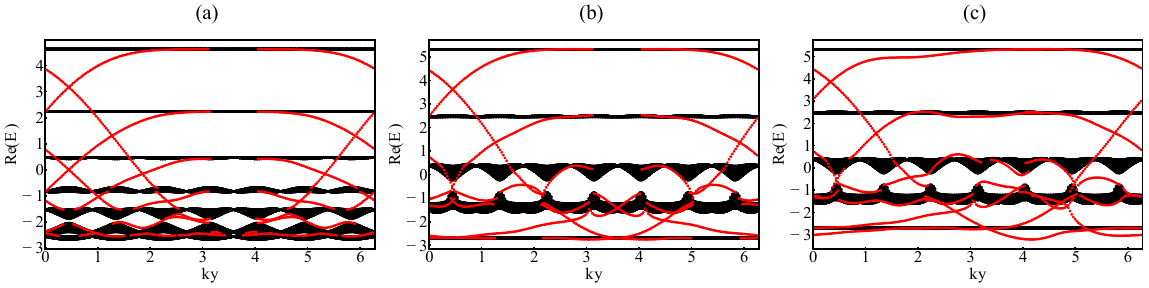}
    \caption{The real part of the $k_y$ dispersion of the Hofstadter model in Eq.~\eqref{eq:nH_triangle_lattice_model} exhibits clear bulk (black) and edge (red) contributions. The perpendicular magnetic field equals $p/q=1/7$ magnetic flux quantum per lattice plaquette, irrespective of the non-reciprocity and the perturbations. Counting the edge modes (red) gives a straightforward identification of the bulk bands (black) topological characters. The non-reciprocity parameter is: (a) $\delta=0$ for a Hermitian system, (b) $\delta=0.5$ for a non-Hermitian system, and (c) $\delta=0.5$ for a non-Hermitian system with perturbations $\sum_{y}(V_2c_{1,y}^\dagger c_{1,y}-V_2c_{2,y}^\dagger c_{2,y})$, $V_2=1$ on the left boundary. }
    \label{fig:nH_Hofstadter_fig2}
\end{figure*}

It is well known that the Hofstadter butterfly is an example of the integer quantum Hall effect, or equivalently, a 2D Chern insulator~\cite{dean2013hofstadter,Tran2015}. Such topological phases host characteristic chiral edges on the open boundaries, which offer physical and straightforward characterizations of the corresponding topological characters - the Chern numbers of the bulk bands~\cite{Hatsugai1993}. Naturally, the recurrence method's capacity for OBC provides a natural arena for exploiting such bulk-edge correspondence. 

For example, we examine the edge modes of the non-Hermitian Hofstadter models, i.e., Eq.~\eqref{eq:nH_triangle_lattice_model} in the presence of the perpendicular magnetic field, especially their dispersions with respect to $k_y$ - the momentum along the open boundaries. We show example results for $p/q=1/7$ in Fig.~\ref{fig:nH_Hofstadter_fig2}, which exhibit apparent edge modes (red), obtained from condition (ii) of the recurrence method, traversing the gaps between the bulk bands (black), mainly obtained from condition (i) of the recurrence method. Interestingly, compared with the Hermitian case in Fig.~\ref{fig:nH_Hofstadter_fig2}(a), although the non-reciprocity $\delta=0.5$ essentially deforms the bulk and edge spectra, the presence of the topological edge modes persists in the non-Hermitian Hofstadter models, as illustrated in Fig.~\ref{fig:nH_Hofstadter_fig2}(b). This is smoking-gun evidence that the corresponding topological phases generalize to non-Hermitian scenarios and, simultaneously, offers a clear-cut identification of the corresponding Chern numbers as $+2$ \footnote{The separation between the lowest two bulk bands is unclear. Thus, we present their topological characterization as a whole. }, $+1$, $-1$, $-1$, $-1$ of the bulk bands from the bottom to the top, as conventionally in Hermitian systems, the (signed) number of chiral edge modes traversing a gap determines the total topological invariants across all bands below the gap. 

Further, we may demonstrate that such topological invariants and phases are robust against perturbations straightforwardly via the recurrence method. For concreteness, we consider adding a boundary perturbation to the original non-Hermitian Hofstadter model: 
\begin{equation}
    H_{\text{nHH}}'=H_{\text{nHH}}+\sum_{y}(V_2c_{1,y}^\dagger c_{1,y}-V_2c_{2,y}^\dagger c_{2,y}), \label{eq:nH_triangle_model_perturbed}
\end{equation}
where the perturbation is translation invariant along $\boldsymbol{y}$, and thus transforms into $V_2 c_{1,k_y}^\dagger c_{1,k_y}-V_2 c_{2,k_y}^\dagger c_{2,k_y}$ in parallel to Eq.~\eqref{eq:nH_complex_magnetic_field} in the $k_y$ basis. Correspondingly, we modify and address the initial conditions with respect to the recurrence relations. 

Specifically, we derive the revised $D_0'$ following the calculations of $D_q'$ and $D_{2q}'$, and the recurrence relation $D_{2q}'(E)=p_q(E)D_q'(E)-TD_0'(E)$, where $T\equiv \prod_{x=1}^qt_xt_x'$. Hence, the bulk spectra remain intact, while the edge spectra instead follow a modified version of Eq.~\eqref{eq:Dk_pk_t}: 
\begin{equation}
    D_q'(E)^2-D_q'(E)D_0'(E)p_q(E)+TD_0'(E)^2=0, \label{eq:Dq_pq_t_perturbed}
\end{equation}
together with one of the following two constraints:
\begin{eqnarray}
    D_q'(E)&=&z_{q,+}(E)D_0'(E), \quad |z_{q,+}(E)|<|z_{q,-}(E)|, \nonumber \\
    D_q'(E)&=&z_{q,-}(E)D_0'(E), \quad |z_{q,-}(E)|<|z_{q,+}(E)|, \label{eq:edge_OBC_condition_perturbed_Hofstadter}
\end{eqnarray}
where $z_{q,\pm}=[p_q(E)\pm\sqrt{p_q(E)^2-4T}]/2$. The results for $V_2=1$ are presented in Fig.~\ref{fig:nH_Hofstadter_fig2}(c). Despite the perturbation, the numbers of edge modes crossing the bulk band gaps and, thus, the implied topological Chern numbers of the bulk bands, remain constant.

\section{Conclusion} \label{sec:conclusion}

We have proposed a novel recurrence method for the energy spectra of non-Hermitian systems under open boundary conditions. Based on the recurrence relations of their characteristic polynomials, our formalism transforms the eigenvalue problems of the target multi-band non-Hermitian Hamiltonians to the solutions of a series of polynomials. It is generally applicable to non-Hermitian systems with (partial) lattice translation symmetries. Therefore, it resolves the accuracy and instability issues in the numerical diagonalization of non-Hermitian matrices, especially for extensive system sizes $N$; with the fast-converging Durand-Kerner method, it also displays better computational complexity than the non-Bloch band theory and the GBZ method, especially for non-Hermitian systems with a considerable period or unit cell with many sublattices. Indeed, as we have demonstrated, it yields consistent analytical spectra or expressions for the non-Hermitian Su-Schrieffer-Heeger models and the non-Hermitian Rice-Mele models under open boundary conditions and efficient and precise numerical results of the non-Hermitian Hofstadter butterfly with a hefty $q_{\text{max}}$ and thus a large magnetic unit cell beyond previous approaches. These spectra host rich physics regarding the PT-symmetry breaking and topological phenomena of the target non-Hermitian systems. 

The recurrence method also provides a targeted formulation of the non-Hermitian edge spectra, with algorithmic and cost advantages over the previous edge matrix method. We have shown its applications on the edge spectra of multi-band non-Hermitian models with nearest-neighbor hopping and open boundary conditions. For instance, we have studied topological edge modes' presence and stability (against perturbations) in the non-Hermitian 1D Rice-Mele and the 2D non-Hermitian Hofstadter models, consistent with and convenient for their topological characterizations - also as examples of the bulk-boundary correspondences. In particular, we have also generalized the size-parity effect of the Hermitian Rice-Mele models to non-Hermitian counterparts and without the generalized chiral symmetry. 

The interplay between open boundaries and non-Hermitian, non-reciprocal physics has gathered much recent interest and attention and witnessed the discoveries of a series of intriguing and exotic phenomena; the introductions of extra factors such as topological phenomena, magnetic fields, dimensionality, etc., has further developed and enriched these research areas; therefore, our general, efficient, and accurate recurrence method offers a timely and promising formalism to facilitate and streamline studies on such non-Hermitian systems under open boundary conditions. The codes and results of this paper are available at Github \footnote{For codes and corresponding results in the Hofstadter model see \url{https://github.com/chenhaoyan05/recurrence_method}}.

\emph{Acknowledgment:} We thank Boris Shapiro for the insightful discussions. We also acknowledge support from the National Natural Science Foundation of China (No.12174008 \& No.92270102) and the National Key R\&D Program of China (No.2022YFA1403700).

\newpage

\bibliography{references}

\clearpage
\appendix
\widetext

\renewcommand\thefigure{S\arabic{figure}}\setcounter{figure}{0}
\renewcommand\thetable{S\arabic{table}}\setcounter{table}{0}

\begin{center}
\textbf{\large Supplementary Materials}\end{center}

\section{Theorems and Algorithms for Recurrence Relations}

\subsection{Properties of Recurrence Relations}

\begin{theorem} \label{sm_theorem_1}
    Let $M=[h_{ab}]$ be the Hamiltonian matrix of model $H=\sum\limits_{ab}c_a^\dagger h_{ab}c_b$ with hopping range $(l,r)$, namely $h_{ab}\neq0$ only for $-r\leq a-b\leq l$. Denote the determinant of submatrix containing first $n$ lattice sites as $M_n=\det[h_{ab}]_{1\leq a,b\leq n}$. Then for large $n$ $M_n$ satisfies a recurrence relation $M_n=\sum\limits_{j=1}^s R_{nj}(\{h_{ab}\})M_{n-j}$ where $s=(l+r)!/l!r!$ is the order of recurrence relation and coefficients $R_{nj}$ satisfy $R_{nj}(\{xh_{ab}\})=x^jR_{nj}(\{h_{ab}\})$.
\end{theorem}

\begin{proof}
    Denote the submatrix of $H$ containing creation operators $c_{i_1}^\dagger,c_{i_2}^\dagger, \cdots, c_{i_N}^\dagger$ and annihilation operators $c_{j_1},c_{j_2},\cdots,c_{j_N}$ as $M_{(i_1,i_2,\cdots, i_N)}^{(j_1,j_2,\cdots, j_N)}$. Then define the determinant of the following submatrix for $1\leq j_1<j_2<\cdots<j_l\leq l+r$:
    \begin{equation}
        D_{n}^{(j_1,j_2,\cdots, j_l)}=\det M_{(1,2,\cdots,n)}^{(1,2,\cdots,n-l,n-l+j_1,n-l+j_2,\cdots,n-l+j_l)}.
    \end{equation}

    We can deduce the recurrence relation of $D_n^{(j_1,j_2,\cdots,j_l)}$ via expanding it with respect to the last creation operator $c_n^\dagger$. Depending on $j_l$, the results are classified into two types:
    \begin{equation}
        D_{n}^{(j_1,j_2,\cdots, j_l)}=\left\{
        \begin{array}{lc}
            h_{n,n+r}D_{n-1}^{(1,j_1+1,\cdots,j_{l-1}+1)}, & \quad j_l=l+r; \\
            (-1)^lh_{n,n-l}D_{n-1}^{(j_1+1,\cdots,j_{l}+1)}+\sum\limits_{i=1}^l(-1)^{l+i}h_{n,n-l+j_i}D_{n-1}^{(1,j_1+1,\cdots,j_{i-1}+1,j_{i+1}+1,\cdots,j_l+1)}, & \quad j_l<l+r.
        \end{array}
        \right. \label{eq:sm_D_n^j}
    \end{equation}

    The upper index of $D_n$ in Eq.~\eqref{eq:sm_D_n^j} can be rearranged using the index of subsets. Denote the $l$-subsets of $\{1,2,\cdots,l+r\}$ as $A_1,A_2,\cdots,A_s$, for example $A_1=\{1,2,\cdots,l\}$ and $A_i=\{j_1,j_2,\cdots,j_l\}$, then $D_{n}^{(j_1,j_2,\cdots, j_l)}$ can be written as:
    \begin{equation}
        D_n^{(i)}\equiv D_{n}^{(j_1,j_2,\cdots, j_l)}. 
    \end{equation}

    Under this definition $D_n^{(1)}=D_n^{(1,2,\cdots,l)}=\det M_{(1,2,\cdots,n)}^{(1,2,\cdots,n)}=M_n$. Then, we can write Eq.~\eqref{eq:sm_D_n^j} in a clear manner: 
    \begin{equation}
        D_n^{(i)}=\sum\limits_{m=1}^s c_{i,m}(n)D_{n-1}^{(m)}, \quad 1\leq i\leq s, \label{eq:sm_D_n^i_recurrence}
    \end{equation}
    where the coefficients $c_{i,m}(n)$ are one of the matrix element $h_{ab}$ multiply by $+1$ or $-1$ (may be $0$ if condition $-r\leq a-b\leq l$ is not satisfied). Let $n=n,n-1,\cdots,n-s+1$ in Eq.~\eqref{eq:sm_D_n^i_recurrence}, we obtain $s^2$ recurrence relations:
    \begin{equation}
        D_{n-j+1}^{(i)}=\sum\limits_{m=1}^s c_{i,m}(n-j+1)D_{n-j}^{(m)}, \quad 1\leq i\leq s, \, 1\leq j\leq s. \label{eq:sm_D_n-j^i_recurrence}
    \end{equation}

    These linearly independent recurrence relations have $s(s+1)$ unknown variables: $D_{n-j}^{(i)}$ for $1\leq i\leq s, \, 0\leq j\leq s$. It follows that using Gaussian elimination, we can eliminate the $s^2-1$ variables $D_{n-j}^{(i)}$, $2\leq i\leq s, \, 0\leq j\leq s$, leaving just one recurrence relation for $D_{n-j}^{(1)}=M_{n-j}$, $0\leq j\leq s$. Write this recurrence relation as 
    \begin{equation}
        M_n=\sum\limits_{j=1}^s R_{nj}(\{h_{ab}\})M_{n-j}. \label{eq:sm_Mn_recurrence}
    \end{equation}
    
    Note that coefficients $c_{i,m}(n-j+1)$ in Eq.~\eqref{eq:sm_D_n-j^i_recurrence} satisfy $c_{i,m}(n-j+1)(\{xh_{ab}\})=xc_{i,m}(n-j+1)(\{h_{ab}\})$, the homogeneous property of coefficients in Eq.~\eqref{eq:sm_Mn_recurrence} $R_{nj}(\{xh_{ab}\})=x^jR_{nj}(\{h_{ab}\})$ is ensured by the Gaussian elimination.
\end{proof}

When the non-Hermitian model is a single-band model with the lattice translational symmetry, namely, we have $h_{ab}=h_{a+1,b+1}$ for any matrix element $h_{ab}$, the coefficients of the recurrence relation $R_{nj}(\{h_{ab}\})$ in Eq.~\eqref{eq:sm_Mn_recurrence} will be independent of $n$: $R_{nj}(\{h_{ab}\})\equiv C_j(\{h_{ab}\})$. One can use the following algorithm to obtain the recurrence relation coefficients $R$ based on Theorem \ref{sm_theorem_1}:

\begin{algorithm}[H]
    \caption{Algorithm for Recurrence Relations}
    \label{alo:recurrence}
    \renewcommand{\algorithmicrequire}{\textbf{Input:}}
    \renewcommand{\algorithmicensure}{\textbf{Output:}}

    \begin{algorithmic}[1]
        \Require $l$, $r$, $n$, $M$    
        \Ensure recurrence relation coefficients $R$  
        \Function{labelsubsets}{$l$,$r$}
            \State labeled $l$-subsets of $\{1,2,\cdots,l+r\}$ such that $S[1]=\{1,2,\cdots,l\}$
            \State \Return $S$
        \EndFunction
        \Function{IndexOf}{$\{j_1,j_2,\cdots,j_l\}$,$l$,$r$}
            \State $S=\text{labelsubsets}(l,r)$
            \State $S[i]=\{j_1,j_2,\cdots,j_l\}$
            \State \Return $i$
        \EndFunction
        \State $S=\text{labelsubsets}(l,r)$
        \State $s=(l+r)!/l!r!$
        \For{$i=1:s$}
            \For{$j=1:s$}
                \State $\text{subset}=S[i]$
                \State $j_l=\text{subset}[l]$
                \If{$j_l==l+r$}
                    \State $\text{subset\_right}=\{1,\text{subset}[1:l-1]+1\}$
                    \State $k=\text{IndexOf}(\text{subset\_right},l,r)$
                    \State $C[i,j,k]=M[n-j+1,n-j+r+1]$
                \Else
                    \State $\text{subset\_right}=\{\text{subset} + 1\}$
                    \State $k=\text{IndexOf}(\text{subset\_right},l,r)$
                    \State $C[i,j,k]=(-1)^lM[n-j+1,n-j-l+1]$
                    \For{$m=1:l$}
                        \State $\text{subset\_right}=\{1,\text{subset}[1:m-1]+1,\text{subset}[m+1:l]+1\}$
                        \State $k=\text{IndexOf}(\text{subset\_right},l,r)$
                        \State $C[i,j,k]=(-1)^{l+m}M[n-j+1,n-j-l+1+\text{subset}[m]]$
                    \EndFor
                \EndIf
            \EndFor
        \EndFor
        \State Find recurrence relation $R[j]$ from $C[i,j,k]$ by Gaussian elimination
    \end{algorithmic}
\end{algorithm}

The OBC spectrum of a single-band non-Hermitian Hamiltonian with the translational symmetry of a lattice vector $H=\sum\limits_{ij}c_i^\dagger h_{ij}c_j$ is determined from solutions of the determinant $\det(E-H)$. Denote the determinant of $N$-site lattice model as $D_N:=\det(EI_N-H_N)$, then $D_N(E)$ is a polynomial of degree $N$ and satisfies the following recurrence relation according to Theorem \ref{sm_theorem_1}:
\begin{equation}
    D_N(E)=\sum\limits_{k=1}^{s}C_k(\{t_{\alpha\beta}\},E)D_{N-k}(E), \label{eq:sm_recursion}
\end{equation}
where the order of the recurrence relation $s=(l+r)!/l!r!$ is related to the hopping range $(l,r)$ of $H$ and coefficients $C_k(\{t_{\alpha\beta}\},E)$ satisfy $C_k(\{xt_{\alpha\beta}\},xE)=x^kC_k(\{t_{\alpha\beta}\},E)$. The characteristic equation of Eq.~\eqref{eq:sm_recursion} is:
\begin{equation}
    z^s=\sum\limits_{k=1}^{s}C_k(\{t_{\alpha\beta}\},E)z^{s-k}. \label{eq:sm_characteristic_equation}
\end{equation}

The determinant $D_N(E)$ is obtained from solutions $z_1(E), z_2(E), \cdots, z_s(E)$ of Eq.~\eqref{eq:sm_characteristic_equation} by:
\begin{equation}
    D_N(E)=\sum\limits_{i=1}^sd_i(E)z_i(E)^N, \label{eq:sm_DN_z}
\end{equation}
where coefficients $d_i(E)$ are derived from the initial conditions by setting $N=0, 1, \cdots, s-1$ and $D_0(E)=1$. OBC spectrum of Hamiltonian $H$ in the thermodynamic limit is equivalent to the solutions of $D_N(E)=0$ in the limit $N\rightarrow\infty$. There is a mathematical theorem states that if polynomials $D_N(E)$ satisfy some non-degenerate conditions, then the $N\rightarrow\infty$ limit of solutions of $D_N(E)=0$ can be determined:
\begin{theorem}
    Let a sequence of polynomials $D_N(E)$ satisfy the recurrence relation in Eq.~\eqref{eq:sm_recursion} and no recurrence relation of an order less than $s$ exists. If there does not exist a phase factor $e^{i\varphi}$ such that $z_i(E)=e^{i\varphi}z_j(E)$ for some $i\neq j$ always hold, then the $N\rightarrow\infty$ limit of zeros of $D_N(E)$ are given by one of the following statements:
    \begin{enumerate}[(i).]
        \item The two solutions with maximal absolute values have the same moduli $|z_i(E)|=|z_j(E)|\geq |z_{m\neq i,j}(E)|$;
        \item Zero coefficient $d_i(E)=0$ before the maximal modulus solution $|z_i(E)|>|z_{j\neq i}(E)|$ in Eq.~\eqref{eq:sm_DN_z}.
    \end{enumerate} \label{sm_theorem_2}
\end{theorem}

For example, in the case of Hatano-Nelson model $H=\sum\limits_i[(t-\gamma)c_{i+1}^\dagger c_i+(t+\gamma)c_{i}^\dagger c_{i+1}]$, the recurrence relation of $D_N(E)$ is given by:
\begin{equation}
    D_N(E)=ED_{N-1}(E)-(t^2-\gamma^2)D_{N-2}(E). \label{eq:sm_HN_recurrence}
\end{equation}

The characteristic equation of recurrence relation in Eq.~\eqref{eq:sm_HN_recurrence} has solutions: 
\begin{equation}
    z_1(E)=\frac{E+\sqrt{E^2-4(t^2-\gamma^2)}}{2}, \quad z_2(E)=\frac{E-\sqrt{E^2-4(t^2-\gamma^2)}}{2}.
\end{equation}

According to condition (i) in Theorem~\ref{sm_theorem_2}, $|z_1(E)|=|z_2(E)|$ gives the expression for the OBC bulk spectra of Hatano-Nelson model:
\begin{equation}
    E=2\sqrt{t^2-\gamma^2}\cos\theta, \quad \theta\in[0, \pi]. 
\end{equation}

\subsection{Recurrence relations in the multiband case}

In the case of multiband Hamiltonians, the coefficients $R_j$ of the recurrence relation are $N$-dependent. For example, in the non-Hermitian SSH model: 
\begin{equation}
    H=\sum\limits_i[(u_1+\gamma_1)a_{i}^\dagger b_i+(u_1-\gamma_1)b_i^\dagger a_i]+\sum\limits_i[(u_2+\gamma_2)a_{i+1}^\dagger b_i+(u_2-\gamma_2)b_i^\dagger a_{i+1}]. \label{eq:sm_nH_SSH}
\end{equation}

The recurrence relations have period $k=2$:
\begin{eqnarray}
    D_{2N}&=&ED_{2N-1}-(u_1^2-\gamma_1^2)D_{2N-2}, \nonumber\\
    D_{2N-1}&=&ED_{2N-2}-(u_2^2-\gamma_2^2)D_{2N-3}. 
\end{eqnarray}

These recurrence relations can be rearranged to the $N$-independent recurrence relation about $D_{2N}$, $D_{2N-2}$, $D_{2N-4}$:
\begin{equation}
    D_{2N}=(E^2+\gamma_1^2+\gamma_2^2-u_1^2-u_2^2)D_{2N-2}-(u_1^2-\gamma_1^2)(u_2^2-\gamma_2^2)D_{2N-4}. 
\end{equation}

One can show that this $N$-independent recurrence relation always exists. We will discuss the recurrence relations of periodic nearest-neighbor hopping models later. For general $m$-band short-range Hamiltonians, we have the following theorem: 
\begin{theorem}
    Consider the Hamiltonian matrix $M=[h_{ab}]$ of non-Hermitian $m$-band model $H=\sum\limits_{ab}c_a^\dagger h_{ab}c_b$ with hopping range $(l,r)$, namely $h_{a,b}=h_{a+m,b+m}$. Denote the determinant of submatrix containing first $n$ sites as $D_n=\det[h_{ab}]_{1\leq a,b\leq n}$. Then for large $n$, $D_n$ satisfies a $n$-independent recurrence relation $D_n=\sum\limits_{j=1}^s C_j(\{h_{ab}\})D_{n-jm}$ where $s=(l+r)!/l!r!$ is the order of recurrence relation and coefficients $C_j$ satisfy $C_j(\{xh_{ab}\})=x^{jm}C_j(\{h_{ab}\})$.
\end{theorem}

This theorem reveals the superiority of the recurrence method in determining the energy spectra of multi-band non-Hermitian systems under OBCs: whatever the number of sites in one unit cell of the non-Hermitian Hamiltonian is, we can always establish a recurrence relation with constant coefficients for its characteristic polynomial $D_N(E)$. The order of this recurrence relation depends solely on the hopping ranges and is not influenced by the number of bands. Consequently, this endows the recurrence method with a distinct advantage for multi-band non-Hermitian Hamiltonians such as the non-Hermitian Hofstadter models. Such an advantage will also be demonstrated more clearly in the following sections.

\section{Recurrence method for one-dimensional periodic nearest-neighbor hopping model} \label{appendix:1d_nn_hopping}

Consider the one-dimensional NN hopping model under OBC with $k$-periodic onsite potential $V_l=V_{l+k}$ and hopping $t_{l}=t_{l+k}$, $t_{l}'=t_{l+k}'$: 
\begin{equation}
    H=\sum\limits_{l=1}^{N-1}(t_{l}c_{l+1}^\dagger c_l+t_{l}'c_l^\dagger c_{l+1}+V_ln_l).  \label{eq:sm_periodic_HN}
\end{equation}

The recurrence relation of the characteristic polynomial $\det(E-H)$ is:
\begin{equation}
    D_N=(E-V_N)D_{N-1}-t_{N-1}t_{N-1}'D_{N-2}. 
\end{equation}

To obtain the size-independent recurrence relation, we write the recurrence relation in the following matrix form:
\begin{equation}
    \begin{pmatrix}
        D_{Nk} \\
        D_{Nk-1}
    \end{pmatrix}=\begin{pmatrix}
        E-V_{k} & -t_{k-1}t_{k-1}' \\
        1 & 0
    \end{pmatrix}\begin{pmatrix}
        D_{Nk-1} \\
        D_{Nk-2}
    \end{pmatrix}\triangleq U_k\begin{pmatrix}
        D_{Nk-1} \\
        D_{Nk-2}
    \end{pmatrix},
\end{equation}
where $U_k$ represents the transfer matrix from $(D_{Nk-1},D_{Nk-2})^T$ to $(D_{Nk},D_{Nk-1})^T$. Applying this transfer matrix recursively, we obtain: 
\begin{equation}
    \begin{pmatrix}
        D_{Nk} \\
        D_{Nk-1}
    \end{pmatrix}=U_kU_{k-1}\cdots U_1\begin{pmatrix}
        D_{(N-1)k} \\
        D_{(N-1)k-1}
    \end{pmatrix}\triangleq U\begin{pmatrix}
        D_{(N-1)k} \\
        D_{(N-1)k-1}
    \end{pmatrix}, \label{eq:sm_overall _recursion}
\end{equation}
where the overall transfer matrix from system size $N-1\rightarrow N$ is:  
\begin{equation}
    U=\begin{pmatrix}
        E-V_{k} & -t_{k-1}t_{k-1}' \\
        1 & 0
    \end{pmatrix}\begin{pmatrix}
        E-V_{k-1} & -t_{k-2}t_{k-2}' \\
        1 & 0
    \end{pmatrix}\cdots \begin{pmatrix}
        E-V_{1} & -t_{k}t_{k}' \\
        1 & 0
    \end{pmatrix}\triangleq\begin{pmatrix}
        A(E) & B(E) \\
        C(E) & D(E)
    \end{pmatrix},
\end{equation}

We can then work out the recurrence relation about $D_{Nk},D_{(N-1)k},D_{(N+1)k}$. In fact, replacing $N\rightarrow N+1$ in the first row equation of Eq.~\eqref{eq:sm_overall _recursion} and multiplying the second row equation with $B(E)$, we obtain:
\begin{eqnarray}
    D_{(N+1)k}&=&A(E)D_{Nk}+B(E)D_{Nk-1},  \label{subeq:1}\\
    B(E)D_{Nk-1}&=&B(E)C(E)D_{(N-1)k}+B(E)D(E)D_{(N-1)k-1}, \label{subeq:2}
\end{eqnarray}

Replacing the terms in Eq.~\eqref{subeq:2} containing $D_{Nk-1},D_{(N-1)k-1}$ using Eq.~\eqref{subeq:1}, we have:
\begin{equation}
    D_{(N+1)k}=\text{Tr}(U)D_{Nk}-\det (U)D_{(N-1)k}. \label{eq:sm_trU_detU}
\end{equation}

Note that both $\text{Tr}(U)$ and $\det(U)$ are size-independent, so we can apply the recurrence relation method in the main text to obtain the OBC spectrum. For $k=3$ case with identical nearest-neighbor hopping $t_{1}=t_{2}=t_{3}\equiv u_1$, $t_{1}'=t_{2}'=t_{3}'\equiv u_2$, the coefficients of the recurrence relation in Eq.~\eqref{eq:sm_trU_detU} are: $\text{Tr}(U)=E^3-E^2(V_1+V_2+V_3)+ 
E(-3u_1u_2+V_1V_2+V_1V_3+V_2V_3)+ u_1u_2(V_1 + V_2 + V_3)-V_1V_2V_3$, $\det(U)=u_1^3u_2^3$. For a general $k$, we can derive $\text{Tr}(U)$ and $\det(U)$ via the following matrix identity:
\begin{equation}
    \begin{pmatrix}
        x_1 & 1 \\
        y_1 & 0
    \end{pmatrix}\begin{pmatrix}
        x_2 & 1 \\
        y_2 & 0
    \end{pmatrix}\cdots\begin{pmatrix}
        x_n & 1 \\
        y_n & 0
    \end{pmatrix}=\begin{pmatrix}
        f_n(x_1,\cdots,x_n,y_1,\cdots,y_n) & f_{n-1}(x_1,\cdots,x_{n-1},y_1,\cdots,y_{n-1}) \\
        y_1f_{n-1}(x_2,\cdots,x_n,y_2,\cdots,y_n) & y_1f_{n-2}(x_2,\cdots,x_{n-1},y_2,\cdots,y_{n-1})
    \end{pmatrix}, \label{eq:sm_matrix_identity}
\end{equation}
where the bivariate polynomials $f_n(x_1,\cdots,x_n,y_1,\cdots,y_n)$ satisfy the following recurrence relation:
\begin{equation}
    f_n(x_1, \ldots, x_n, y_1, \ldots, y_n)=f_{n-1}(x_1, \ldots, x_{n-1}, y_1, \ldots, y_{n-1}) x_n+f_{n-2}(x_1, \ldots, x_{n-2}, y_1, \ldots, y_{n-2}) y_n. \label{eq:sm_fn}
\end{equation}

Let us define $f_{-1}=0,f_0=1,f_i=f_i(x_1, \ldots, x_i, y_1, \ldots, y_i), f_i^{\prime}=f_i(x_2, \ldots, x_{i+1}, y_2, \ldots, y_{i+1})$, we can readily prove Eq.~\eqref{eq:sm_fn} by induction, since:
\begin{equation}
    \begin{pmatrix}
    f_{n-1} & f_{n-2} \\
    y_1 f_{n-2}^{\prime} & y_1 f_{n-3}^{\prime}
    \end{pmatrix}\begin{pmatrix}
    x_n & 1 \\
    y_n & 0
    \end{pmatrix}=\begin{pmatrix}
    f_n & f_{n-1} \\
    y_1 f_{n-1}^{\prime} & y_1 f_{n-2}^{\prime}
    \end{pmatrix}.
\end{equation}

Applying Eq.~\eqref{eq:sm_matrix_identity} to $\text{Tr}(U^T)$, we have:
\begin{equation}
    \text{Tr}(U^T)=f_k(E-V_1,\cdots,E-V_k,-t_kt_k',-t_1t_1',\cdots,-t_{k-1}t_{k-1}')-t_kt_k'f_{k-2}(E-V_2,\cdots,E-V_{k-1},-t_1t_1',\cdots,-t_{k-2}t_{k-2}').
\end{equation}

From Eq.~\eqref{eq:sm_fn}, one finds:
\begin{align}
    f_k(E-V_1,\cdots,E-V_k,-t_kt_k',-t_1t_1',\cdots,-t_{k-1}t_{k-1}')&=\Delta_{1,k}, \nonumber \\
    f_{k-2}(E-V_2,\cdots,E-V_{k-1},-t_1t_1',\cdots,-t_{k-2}t_{k-2}')&=\Delta_{2,k-1},
\end{align}
where:
\begin{equation}
    \Delta_{i,j}=\left\{
    \begin{array}{ccl}
        \det\begin{pmatrix}
        E-V_i & 1 & &  \\
        t_{i}t_{i}' & E-V_{i+1} & 1 &  \\
        & \ddots & \ddots & 1  \\
        & &t_{j-1}t_{j-1}' & E-V_j \\
    \end{pmatrix} & & j>i-1, \\
        1 & & j=i-1, \\
        0 & & j<i-1.
    \end{array}\right.
    \label{eq:sm_Delta}
\end{equation} 

Thus, the recurrence relation~in Eq.\eqref{eq:sm_trU_detU} is written as:
\begin{equation}
    D_{(N+1)k}=(\Delta_{1,k}-t_{k}t_{k}'\Delta_{2,k-1})D_{Nk}-\prod_{j=1}^kt_{j}t_{j}'D_{(N-1)k},  \label{eq:sm_size_independent_recurrence}
\end{equation}
which gives the recurrence relation in the main text.

Defining $p_k(E)\equiv\Delta_{1,k}-t_{k}t_{k}'\Delta_{2,k-1}$ and $t\equiv \prod_{j=1}^kt_{j}t_{j}'$, then applying the recurrence method in the main text, we have the solutions of characteristic equation for recurrence relation Eq.~\eqref{eq:sm_size_independent_recurrence}: $x_\pm=\frac{1}{2}(p_k\pm\sqrt{p_k^2-4t})$. From condition (i): $|x_+|=|x_-|$, we obtain the OBC bulk spectrum of Eq.~\eqref{eq:sm_periodic_HN} satisfying the relation: 
\begin{equation}
    p_k(E)=2\sqrt{t}\cos\theta. \label{eq:sm_bulk_OBC_spectrum}
\end{equation}

For the $k=4$ case, $p_4(E)=E^4+c_3E^3+c_2E^2+c_1E+c_0$, where: 
\begin{eqnarray}
c_3&=&-(V_1+V_2+V_3+V_4), \nonumber\\
c_2&=&-T_1-T_2-T_3-T_4+V_1V_2+V_1V_3+V_2V_3+V_1V_4+V_2 V_4+V_3V_4, \nonumber\\
c_1&=&T_1(V_3+V_4)+T_2(V_1+V_4)+T_3(V_1+V_2)+T_4(V_2+V_3)-V_1V_2V_3-V_1V_2V_4-V_1V_3V_4-V_2V_3V_4, \nonumber\\
c_0&=&T_1T_3+T_2T_4-T_1V_3V_4-T_2V_1V_4-T_3V_1V_2-T_4V_2V_3,
\end{eqnarray}
and $T_i\equiv t_{i}t_{i}'$. 

For the edge spectra from condition (ii), we need to consider the initial conditions of the characteristic polynomial $D_{Nk}(E)=d_+(E)x_+(E)^N+d_-(E)x_-(E)^N$. Using $D_0=1$ and condition (ii) $d_+(E)=0$ or $d_-(E)=0$, we obtain that the edge spectrum is the set: 
\begin{equation}
    \{E|D_k(E)=x_+(E),\,|x_+(E)|<|x_-(E)|\}\cup\{E|D_k(E)=x_-(E),\,|x_+(E)|>|x_-(E)|\}.  \label{eq:sm_OBC_edge_spectrum}
\end{equation}

\section{Numerical Calculations for the OBC spectrum of the $k$-periodic Hatano-Nelson model}

It is widely known that direct diagonalization of the non-Hermitian OBC Hamiltonian matrix $H_{\text{OBC}}$ is sensitive to the matrix dimension, and accumulation of the digital errors may result in incorrect or unstable eigenvalues~\cite{Zhesen2020}. Based on the generalized Brillouin zone (GBZ) or recurrence method, we can develop numerical algorithms to resolve the issue of direct diagonalization. 

\subsection{Numerical calculation based on the Generalized Brillouin zone method}

Here, we present an algorithm applicable to single-band non-Hermitian Hamiltonians using the GBZ method~\cite{Chen2023}. The characteristic equation of non-Bloch Hamiltonian $\det[H(\beta)-E]=0$ is now written as: 
\begin{equation}
    E=H(\beta)=\sum\limits_{n=-p}^q h_n\beta^n. \label{eq:sm_single_band}
\end{equation}
The OBC spectrum $\sigma_{\text{OBC}}$ in the thermodynamic limit is given by the condition $|\beta_p(E)|=|\beta_{p+1}(E)|$, in which $\beta_p(E_{\text{OBC}})$ and $\beta_{p+1}(E_{\text{OBC}})$ constitute the GBZ. We seek for $M$ points $E_1, E_2, \cdots, E_M$ on the OBC spectrum $\sigma_{\text{OBC}}$, so we assume that:
\begin{equation}
    \beta_p(E_l)=\beta_le^{i\theta_l}, \quad \beta_{p+1}(E_l)=\beta_{l}e^{-i\theta_l}, \label{eq:sm_beta_p}
\end{equation}
where $\beta_l$ is complex and the phases $\theta_l=l\pi/(M+1)$, $l=1,2,\cdots,M$ characterize the azimuth difference between $\beta_p$ and $\beta_{p+1}$. This choice of the phase factor ensures an approximately uniform distribution of $\beta_p/\beta_{p+1}$ over the GBZ. Substituting Eq.~\eqref{eq:sm_beta_p} into Eq.~\eqref{eq:sm_single_band}, we obtain: 
\begin{eqnarray}
    E&=&H(\beta_p)=\sum\limits_{n=-p}^q h_n\beta_l^ne^{in\theta_l}, \nonumber\\
    E&=&H(\beta_{p+1})=\sum\limits_{n=-p}^q h_n\beta_l^ne^{-in\theta_l}. \label{eq:sm_E_subeq2}
\end{eqnarray}
Eliminating $E$ in Eq.~\eqref{eq:sm_E_subeq2}, we obtain the following equation for $\beta_l$:
\begin{equation}
    \sum\limits_{n=-p}^q h_n\beta_l^n\sin n\theta_l=0. \label{eq:sm_eq_beta_l}
\end{equation}

The algorithmic procedure for evaluating $\beta_l$ and thus $E=H(\beta_le^{\pm i\theta_l})$ is summarized as follows: 
\begin{enumerate}
    \item Start from $l=1$, solve Eq.~\eqref{eq:sm_eq_beta_l} and obtain $p+q$ roots $\beta_l^{(i)}$.
    \item For each $\beta_l^{(i)}$, obtain $E_l^{(i)}$ from Eq.~\eqref{eq:sm_E_subeq2}. Let $E=E_l^{(i)}$ in Eq.~\eqref{eq:sm_single_band}, solve this equation to determine the roots sorted by their moduli: $|\beta_1(E_l^{(i)})|\leq|\beta_2(E_l^{(i)})|\leq\cdots\leq|\beta_{p+q}(E_l^{(i)})|$.
    \item Check whether $|\beta_l^{(i)}|=|\beta_p(E_l^{(i)})|=|\beta_{p+1}(E_l^{(i)})|$. If true, then $E_l^{(i)}\in \sigma_{\text{OBC}}$ and add $E_l^{(i)}$ to the spectrum list; otherwise $E_l^{(i)}\notin \sigma_{\text{OBC}}$. Return to step 2 until all $\beta_l^{(i)}$ have been checked.
    \item $l\rightarrow l+1$ and back to step 1 until $l=M$.
\end{enumerate}

\subsection{Numerical calculation based on the recurrence relation method} \label{appendix:numerical_recurrence_method}

From Appendix~\ref{appendix:1d_nn_hopping}, we know that OBC bulk spectrum of non-Hermitian models in Eq.~\eqref{eq:sm_periodic_HN} is given by $p_k(E)=\Delta_{1,k}-T_k\Delta_{2,k-1}=2\sqrt{t}\cos\theta$, where $\Delta_{i,j}$ is given in Eq.~\eqref{eq:sm_Delta}, and $T_j\equiv t_{j}t_{j}'$, $t=\prod_{j=1}^k T_j$. To obtain the OBC spectrum, we have to derive the coefficients in the polynomial $p_k(E)=E^k+C_1E^{k-1}+\cdots+C_{k-1}E+C_k$. $\Delta_{1,k}$ is a polynomial of degree $k$ and satisfies the recurrence relation:
\begin{equation}
    \Delta_{1,n}=(E-V_n)\Delta_{1,n-1}-T_{n-1}\Delta_{1,n-2},\quad \Delta_{1,0}=1,\quad \Delta_{1,1}=E-V_1, \quad 2\leq n\leq k. \label{eq:sm_Delta_recurrence}
\end{equation}

We identify the polynomial coefficients:
\begin{equation}
    \Delta_{1,k}(E)=E^k+c_1E^{k-1}+\cdots+c_{k-1}E+c_k, \quad \Delta_{1,n}(E)=E^n+c_1^{(n)}E^{n-1}+\cdots+c_{n-1}^{(n)}E+c_n^{(n)}, \quad 1\leq n\leq k
\end{equation}
where $c_i^{(k)}=c_i$. Using the relation Eq.~\eqref{eq:sm_Delta_recurrence}, we can calculate the coefficients $c_i^{(j)}$ recursively. For $n\geq 3$: 
\begin{eqnarray}
    c_1^{(n)}&=&c_1^{(n-1)}-V_n, \nonumber\\
    c_2^{(n)}&=&c_2^{(n-1)}-V_nc_1^{(n-1)}-T_{n-1}, \nonumber\\
    c_j^{(n)}&=&c_j^{(n-1)}-V_nc_{j-1}^{(n-1)}-T_{n-1}c_{j-2}^{(n-2)}, \quad 3\leq j\leq n-1, \nonumber\\
    c_n^{(n)}&=&-V_nc_{n-1}^{(n-1)}-T_{n-1}c_{n-2}^{(n-2)}, 
\end{eqnarray}
with initial conditions $c_1^{(1)}=-V_1$, $c_1^{(2)}=-(V_1+V_2)$, $c_2^{(2)}=V_1V_2-T_1$. Similarly,  $\Delta_{2,k-1}$ can also be recursively determined. We may write $\Delta_{2,k-1}$ and $\Delta_{2,n}$ as:
\begin{equation}
    \Delta_{2,k-1}(E)=E^{k-2}+\tilde{c}_1E^{k-3}+\cdots+\tilde{c}_{k-3}E+\tilde{c}_{k-2}, \quad \Delta_{2,n}(E)=E^{n-1}+\tilde{c}_1^{(n)}E^{n-2}+\cdots+\tilde{c}_{n-2}^{(n)}E+\tilde{c}_{n-1}^{(n)}, \, 2\leq n\leq k-1. 
\end{equation}

For $n\geq 4$: 
\begin{eqnarray}
    \tilde{c}_1^{(n)}&=&\tilde{c}_1^{(n-1)}-V_n, \nonumber\\
    \tilde{c}_2^{(n)}&=&\tilde{c}_2^{(n-1)}-V_n\tilde{c}_1^{(n-1)}-T_{n-1}, \nonumber\\
    \tilde{c}_j^{(n)}&=&\tilde{c}_j^{(n-1)}-V_n\tilde{c}_{j-1}^{(n-1)}-T_{n-1}\tilde{c}_{j-2}^{(n-2)}, \quad 3\leq j\leq n-2, \nonumber\\
    \tilde{c}_{n-1}^{(n)}&=&-V_n\tilde{c}_{n-2}^{(n-1)}-T_{n-1}\tilde{c}_{n-3}^{(n-2)}, 
\end{eqnarray}
with initial conditions $\tilde{c}_1^{(2)}=-V_2$, $\tilde{c}_1^{(3)}=-(V_2+V_3)$, $\tilde{c}_2^{(3)}=V_2V_3-T_2$. The coefficients of polynomial $p_k(E)=E^k+C_1E^{k-1}+\cdots+C_{k-1}E+C_k$ are given by: 
\begin{equation}
    C_1=c_1, \quad C_2=c_2-T_k, \quad C_{i}=c_i-T_k\tilde{c}_{i-2}, \quad 3\leq i\leq k. 
\end{equation}

Then, the OBC bulk spectrum of Eq.~\eqref{eq:sm_periodic_HN} in the thermodynamic limit is determined from the equation: 
\begin{equation}
    E^k+C_1E^{k-1}+\cdots+C_{k-1}E+C_k=2\sqrt{t}\cos\theta, \quad \theta\in[0,2\pi].  \label{eq:sm_OBC_spectrum}
\end{equation}

For a finite system with $Nk$ lattice sites, we can approximate the finite-size bulk OBC spectrum by properly choosing $N$ $\theta$ values in Eq.~\eqref{eq:sm_OBC_spectrum}. A convenient choice is: $\theta_i=i\pi/(N+1), i=1,2,\cdots,N$. 

The OBC edge spectrum which originates from condition (ii) in the recurrence method can be determined from Eq.~\eqref{eq:sm_OBC_edge_spectrum}, where $x_\pm(E)=\frac{1}{2}(p_k(E)\pm\sqrt{p_k(E)^2-4t})$ and $D_k$ is just $\Delta_{1,k}$ obtained in the previous text. Solutions of the following polynomial equation: 
\begin{equation}
    \Delta_{1,k}(E)^2-\Delta_{1,k}(E)p_k(E)+t=0,
\end{equation}
which satisfy either of the two conditions: (i). $\Delta_{1,k}(E)=x_+(E)$, $|x_+(E)|<|x_-(E)|$ or (ii). $\Delta_{1,k}(E)=x_-(E)$, $|x_-(E)|<|x_+(E)|$ belong to the OBC edge spectrum of Eq.~\eqref{eq:sm_periodic_HN}.

\section{Non-Hermitian Hofstadter butterfly on the triangular lattice}

To demonstrate the power of the recurrence method, we consider a two-dimensional non-Hermitian Hofstadter model on a triangular lattice, which is beyond the scope of generalized Brillouin zone theory and similarity transformation. The Hamiltonian of the non-reciprocal triangular lattice model is written as: 
\begin{align}
    H_{\text{nHH}}=&\sum\limits_{x,y}(\sqrt{(1-\delta)/(1+\delta)}c_{x+1,y}^\dagger c_{x,y}+\sqrt{(1+\delta)/(1-\delta)}c_{x-1,y}^\dagger c_{x,y}+\sqrt{(1+\delta)/(1-\delta)}c_{x,y+1}^\dagger c_{x,y} \\
    &\sqrt{(1-\delta)/(1+\delta)}c_{x,y-1}^\dagger c_{x,y}+\sqrt{(1+\delta)/(1-\delta)}c_{x+1,y-1}^\dagger c_{x,y}+\sqrt{(1-\delta)/(1+\delta)}c_{x-1,y+1}^\dagger c_{x,y}), \label{eq:sm_nH_triangle_lattice_model}
\end{align}
where $\delta$ represents the non-reciprocity. We set the convention that the product of opposite hopping amplitudes equals $1$, and $\hbar=e=1$. Note that we set our lattice vectors along the sides of the triangles. Applying a uniform magnetic field with Landau gauge $\vec{A}=(0, Bx)$, we can write the Hamiltonian in Eq.~\eqref{eq:sm_nH_triangle_lattice_model} for a particular $k_y$ as:
\begin{align}
    H_{\text{nHH}}(k_y)&=\sum\limits_x [\sqrt{(1-\delta)/(1+\delta)}c_{x+1,k_y}^\dagger c_{x,k_y}+\sqrt{(1+\delta)/(1-\delta)}c_{x-1,k_y}^\dagger c_{x,k_y}+\sqrt{(1+\delta)/(1-\delta)}e^{-iBx}e^{ik_y}c_{x,k_y}^\dagger c_{x,k_y} \notag \\
    &+\sqrt{(1-\delta)/(1+\delta)}e^{iBx}e^{-ik_y}c_{x,k_y}^\dagger c_{x,k_y}+\sqrt{(1+\delta)/(1-\delta)}e^{-ik_y}e^{iB(x+1/2)}c_{x+1,k_y}^\dagger c_{x,k_y} \notag \\
    &+\sqrt{(1-\delta)/(1+\delta)}e^{-iB(x-1/2)}e^{ik_y}c_{x-1,k_y}^\dagger c_{x,k_y}], \label{eq:sm_nH_1d_Hofstadter}
\end{align}
where the magnetic phase factors along the three sides of a triangular plaquette are given by
\begin{eqnarray}
    e^{i\phi_1}&=&\exp\Big(-i\int_{(x,y)}^{(x,y+1)}\Vec{A}\cdot d\Vec{l}\Big)=e^{-iBx}, \nonumber\\
    e^{i\phi_2}&=&\exp\Big(-i\int_{(x,y+1)}^{(x+1,y)}\Vec{A}\cdot d\Vec{l}\Big)=e^{iB(x+1/2)}, \nonumber\\
    e^{i\phi_3}&=&\exp\Big(-i\int_{(x+1,y)}^{(x,y)}\Vec{A}\cdot d\Vec{l}\Big)=1.
\end{eqnarray}

The Hamiltonian in Eq.~\eqref{eq:sm_nH_1d_Hofstadter} can be written as a one-dimensional non-Hermitian Hamiltonian with NN hoppings: 
\begin{equation}
    H_{\text{nHH}}(k_y)=\sum\limits_x(t_xc_{x+1,k_y}^\dagger c_{x,k_y}+t_x'c_{x,k_y}^\dagger c_{x+1,k_y}+V_xc_{x,k_y}^\dagger c_{x,k_y}), \label{eq:sm_nH_1d_q_periodic}
\end{equation}
where the NN hoppings $t_x$, $t_x'$, and the onsite potentials $V_x$ are given by:
\begin{eqnarray}
    t_x&=&\sqrt{(1-\delta)/(1+\delta)}+\sqrt{(1+\delta)/(1-\delta)}e^{-ik_y}e^{iB(x+1/2)}, \nonumber\\
    t_x'&=&\sqrt{(1+\delta)/(1-\delta)}+\sqrt{(1-\delta)/(1+\delta)}e^{ik_y}e^{-iB(x+1/2)}, \nonumber\\
    V_x&=&\sqrt{(1+\delta)/(1-\delta)}\exp[-i(Bx-k_y)]+\sqrt{(1-\delta)/(1+\delta)}\exp[i(Bx-k_y)].
\end{eqnarray}

When the magnetic flux is commensurate, i.e., $B=2\pi (p/q)$ so that the flux per triangle plaquette is $\Phi=\pi (p/q)$, the Hamiltonian in Eq.~\eqref{eq:sm_nH_1d_q_periodic} is a one-dimensional $q$-periodic model with NN hopping $t_x=t_{x+q}$, $t_x'=t_{x+q}'$, and $V_x=V_{x+q}$. The bulk energy spectra under OBC are obtained in Appendix~\ref{appendix:1d_nn_hopping} as the solutions of the following polynomials: 
\begin{equation}
    p_q(E)=2\sqrt{T}\cos\theta, \label{eq:sm_bulk_OBC_spectrum_Hofstadter}
\end{equation}
where $T\equiv \prod_{x=1}^qt_xt_x'$ and $p_q(E)$ is a $q-$th degree polynomial whose coefficients can be recursively determined as shown in Appendix \ref{appendix:numerical_recurrence_method}. To obtain bulk spectra with lattice size $Nq$, we can choose $\theta_i$ as $\theta_i=i\pi/(N+1), i=1, 2, \cdots, N$, then the bulk spectra are given by solving a series of $q$-th polynomials with the variation of constant terms only: $p_q(E)-2\sqrt{T}\cos\theta_i=0$. By applying the recurrence method, we resolve the precision issues in non-Hermitian numerical diagonalization and significantly reduce the time complexity since the period $q$ is usually much smaller than the number of unit cells $N$. 

The edge spectra under OBC are obtained from the equation:  
\begin{equation}
    D_q(E)^2-D_q(E)p_q(E)+T=0,
\end{equation}
with either of the following two constraints satisfied: (i). $D_q(E)=z_+(E)$, $|z_+(E)|<|z_-(E)|$; (ii). $D_q(E)=z_-(E)$, $|z_-(E)|<|z_+(E)|$, where $D_q(E)$ is the determinant of the first $q$-size unit cell and $z_{\pm}(E)=\frac{1}{2}(p_q(E)\pm\sqrt{p_q(E)^2-4T})$. 

The non-reciprocity $\delta$ in the Hamiltonian in Eq.~\eqref{eq:sm_nH_1d_Hofstadter} can be considered as the effect of an imaginary magnetic field. The magnetic phase of the counterclockwise closed path along the edges of the triangular plaquette is: 
\begin{equation}
    e^{i\phi}=\exp\Big(-i\oint \vec{A}\cdot d\vec{l}\Big)=e^{i(\phi_1+\phi_2+\phi_3)}=e^{iB/2}. \label{eq:sm_magnetic_phase_factor}
\end{equation}
We find that when the magnetic field is ``imaginary": $B=iB_i$, the magnetic phase factor in Eq.~\eqref{eq:sm_magnetic_phase_factor} becomes a scaling factor $e^{-B_i/2}$, which is identical to non-reciprocity. Hence, the two-dimensional non-reciprocal triangular lattice model in Eq.~\eqref{eq:sm_nH_1d_Hofstadter} under the magnetic field 
$B$ is equivalent to a Hermitian model under a complex magnetic field $B_c=B_r+iB_i$ if $B_r=B$ and:
\begin{equation}
    \frac{(1-\delta)^{3/2}}{(1+\delta)^{3/2}}=e^{-B_i/2}. 
\end{equation}

While $B_r$ has to be commensurate to facilitate our calculations, the complex magnetic field $B_c$ endows us with more degrees of freedom, where parallel comparisons can be drawn with the complex semiclassical theory. For instance, we can fix the modulus of the complex magnetic field $|B_c|$ and vary its phase angle $\theta$: $B_c=|B_c|(\cos\theta+i\sin\theta)$. Given a specific, commensurate applied magnetic field $B=2\pi(p/q)$, the non-reciprocity parameter $\delta$ satisfies: 
\begin{equation}
    \Big(3\ln\frac{1+\delta}{1-\delta}\Big)^2+(2\pi p/q)^2=|B_c|^2, \label{eq:bcdelta}
\end{equation}
which yields $\theta=\arccos\dfrac{2\pi p}{q|B_c|}$ and $\delta=\tanh\dfrac{|B_c|\sin\theta}{6}$. For example, we present in Fig.~\ref{fig:supp_constant_Bc} the results for $|B_c|=6$, which would otherwise have been incommensurate with the lattice.

\begin{figure}[htb]
    \centering
    \includegraphics[width=0.8\linewidth]{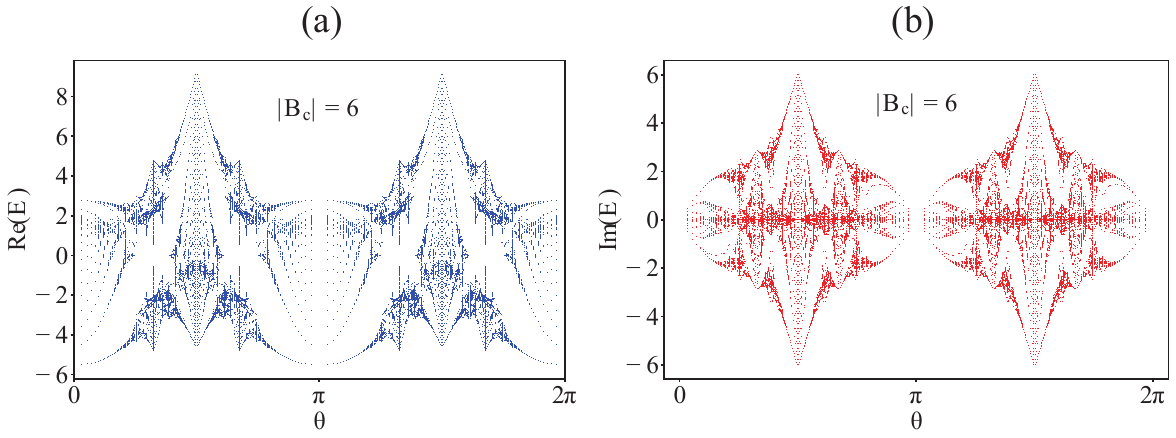}
    \caption{Our recurrence method gives the real (a) and imaginary (b) parts of the OBC spectra of the non-Hermitian Hofstadter model on a triangular lattice with a complex magnetic field $B_c=|B_c|\exp(i\theta)$ of constant modulus $|B_c|=6$ and variable phase angle $\theta$. For each $\theta$, we have applied a commensurate magnetic field $B_r=2\pi p/q$ with the corresponding non-reciprocity $\delta$ as in Eq.~\eqref{eq:bcdelta}. }
    \label{fig:supp_constant_Bc}
\end{figure}

\section{Complex Semiclassical Theory}

In this section, we derive the Landau fan structure near the band top, namely, the Landau level scaling: 
\begin{equation}
    \sqrt{1-\delta^2}E_n=6-3B(n+1/2), \quad n=0, 1, 2, \cdots, \label{eq:sm_Landau_levels}
\end{equation}
for the model in Eq.~\eqref{eq:sm_nH_triangle_lattice_model} in a magnetic field $B=2\pi p/q$. For simplicity, we consider a non-Hermitian Hofstadter model with clockwise hopping $1+\delta$ and counterclockwise hopping $1-\delta$ instead, whose spectrum equals that of Eq.~\eqref{eq:sm_nH_triangle_lattice_model} after the $\sqrt{1-\delta^2}$ factor. In the complex semiclassical theory, we may also treat the non-reciprocity as a non-Hermitian energy dispersion or an imaginary magnetic field. Here, we take the former route for an approximate analytical solution. 

Without the applied magnetic field, the energy spectrum in the momentum space is: 
\begin{equation}
    E=2\cos k_x+2i\delta\sin k_x +2\cos\frac{\sqrt{3}k_y-k_x}{2}+2i\delta\sin\frac{\sqrt{3}k_y-k_x}{2}+2\cos\frac{-\sqrt{3}k_y-k_x}{2}+2i\delta\sin\frac{-\sqrt{3}k_y-k_x}{2}, \label{eq:sm_E_kx_ky}
\end{equation}
which we expand to the quadratic order around the band top $k_x=k_y=0$: 
\begin{equation}
    E\approx 2-k_x^2+2-2\left(\frac{\sqrt{3}k_y-k_x}{2}\right)^2+2-2\left(\frac{-\sqrt{3}k_y-k_x}{2}\right)^2=6-\frac{3}{2}(k_x^2+k_y^2). \label{eq:sm_E_kx_ky_expansion}
\end{equation}

In a perpendicular magnetic field, we have $k_y\rightarrow k_y-Bx$, hence:
\begin{equation}
    E\approx 6-\frac{3}{2}\left[k_x^2+(Bx-k_y)^2\right] = 6-3B(n+1/2), \quad n=0, 1, 2,\cdots, \label{eq:sm_E_kx_x}
\end{equation}
which is exactly solvable, similar to a harmonic oscillator. Putting back the $\sqrt{1-\delta^2}$ due to the convention difference of the hopping parameters, we arrive at the Landau fan structure $\sqrt{1-\delta^2}E_n=6-3B(n+1/2)$ we mentioned in the main text. 

Such analytical approximation is also available to the band bottom. Further away from the band top or bottom, the approximation becomes less precise due to the increasing weights of the higher-order terms, where semiclassical solutions are still available numerically through the quantization condition $\oint \boldsymbol{p}\cdot d\boldsymbol{r}=(n+1/2)h$ upon the complex orbits \cite{Yang2024}.

\section{Durand-Kerner Method}

By applying the recurrence method, we transform the problem of obtaining OBC spectra to the problem of solving a sequence of polynomials in Eq.~\eqref{eq:sm_OBC_spectrum}. Here, we use a numerical method called Durand-Kerner Method, widely used to simultaneously calculate all the complex roots of a univariate polynomial. Given a univariate polynomial $P(x)$ with complex coefficients:
\begin{equation}
    P(x)=a_0x^n+a_1x^{n-1}+\cdots+a_{n-1}x+a_n, \quad a_0\neq 0. 
\end{equation}

In the Durand-Kerner method, we start from the distinct initial trial solutions $x_1^{(0)},x_2^{(0)},\cdots,x_n^{(0)}$ and iteratively improve the trial solutions via the formula: 
\begin{equation}
    x_j^{(i+1)}=x_j^{(i)}-\frac{P(x_j^{(i)})}{a_0\prod_{k=1,j\neq k}^n(x_j^{(i)}-x_k^{(i)})}, \quad j=1, 2, \cdots, n.
\end{equation}

In practice, after a sufficient number of iteration steps, the approximations $x_j^{(N)}$ converge to the roots of $P(x)$ with quadratic convergence rate if the trial solutions $x_j^{(0)}$ are properly chosen. Here, we adopt the trial solutions as $x_j^{(0)}=(0.4+0.9i)^j$.

\end{document}